\documentclass[12pt,reqno]{amsart}
\usepackage{sourceserifpro}
\usepackage{sourcecodepro}
\usepackage[T1]{fontenc}
\usepackage[italic,eulergreek,symbolmisc]{mathastext}
\usepackage{microtype}
\usepackage{natbib}
\usepackage[dvipsnames]{xcolor}
\usepackage{hyperref}
\usepackage{dsfont}
\usepackage{geometry}
\usepackage{setspace}
\usepackage{amssymb, color}
\usepackage{amsmath}
\usepackage{amsfonts}
\usepackage{rotating}
\usepackage{caption}
\usepackage{appendix}
\usepackage{amssymb}
\usepackage{graphicx}
\usepackage{booktabs}
\usepackage{threeparttable}
\usepackage[nomarkers]{endfloat}



\MTsetmathskips{f}{\thinmuskip}{0mu}
\MTsetmathskips{y}{\thinmuskip}{0mu}
\MTsetmathskips{p}{\thinmuskip}{0mu}
\MTsetmathskips{l}{0mu}{\thinmuskip}
\MTsetmathskips{j}{\thinmuskip}{\thinmuskip}
\hypersetup{colorlinks=true,citecolor=Blue,filecolor=Blue,linkcolor=Blue,urlcolor=Blue}
\allowdisplaybreaks[4]
\newtheorem{theorem}{Theorem}[section]

\newtheorem{algorithm}{Algorithm}[section]

\newtheorem{corollary}{Corollary}[section]

\newtheorem{lemma}{Lemma}[section]

\newtheorem{remark}{Remark}[section]

\newtheorem{assumption}{Assumption}[section]

\begin{document}
\title[Logit-Based Alternatives to Two-Stage Least Squares]{Logit-Based Alternatives to
Two-Stage Least Squares}
\author[Chetverikov]{Denis Chetverikov}
\author[Hahn]{Jinyong Hahn}
\author[Liao]{Zhipeng Liao}
\author[Sheng]{Shuyang Sheng}
\address[D. Chetverikov]{Department of Economics, UCLA, Bunche Hall, 8283,
315 Portola Plaza, Los Angeles, CA 90095, USA.}
\email{chetverikov@econ.ucla.edu}
\address[J. Hahn]{Department of Economics, UCLA, Bunche Hall, 8283, 315
Portola Plaza, Los Angeles, CA 90095, USA.}
\email{hahn@econ.ucla.edu}
\address[Z. Liao]{Department of Economics, UCLA, Bunche Hall, 8283, 315
Portola Plaza, Los Angeles, CA 90095, USA.}
\email{zhipeng.liao@econ.ucla.edu}
\address[S. Sheng]{Shenzhen Finance Institute, School of Management and Economics, The Chinese University of Hong Kong, Shenzhen, 2001 Longxiang Boulevard, Shenzhen, Guangdong, 518172, China.}
\email{shengshuyang@cuhk.edu.cn}
\date{\today }

\begin{abstract}
We propose logit-based IV and augmented logit-based IV estimators that serve
as alternatives to the traditionally used 2SLS estimator in the model where
both the endogenous treatment variable and the corresponding instrument are
binary. Our novel estimators are as easy to compute as the 2SLS estimator
but have an advantage over the 2SLS estimator in terms of causal
interpretability. In particular, in certain cases where the probability
limits of both our estimators and the 2SLS estimator take the form of
weighted-average treatment effects, our estimators are guaranteed to yield
non-negative weights whereas the 2SLS estimator is not. 
\end{abstract}

\maketitle

\section{Introduction}

We study the problem of instrumental variable estimation in the setting with
a binary treatment and a binary instrument in the presence of controls.
Numerous parametric and nonparametric instrumental variable estimators have
been proposed in the literature for this setting, and among all of them,
perhaps the most important one is the 2SLS estimator. It is simple to
compute and has straightforward motivation in the case of constant treatment
effects. However, it has been recently demonstrated by \cite{BBMT22} that in
the case of heterogeneous treatment effects, the 2SLS estimator has multiple
issues unless saturated controls\footnote{%
When controls are discrete, the vector of controls is said to be saturated if it consists of dummy
variables such that for all its realizations, one and only one dummy takes
value one.} are being used, which is rarely the case in practice. In this
paper, we propose a new instrumental variable estimator that alleviates some
of the problems of the 2SLS estimator and is as simple to compute as the
2SLS estimator itself.

Like the 2SLS estimator, our estimator consists of two steps. In the first
step, we run a logit regression of the instrument on the set of controls. In
the second step, we use a classic instrumental variable estimator for the
linear regression of the outcome on the treatment using residuals from the
logit regression calculated on the first step as an instrument. We refer to
this procedure as the \emph{logit-based instrumental variable (IV) estimator}%
.

Under the standard monotonicity (no defiers) and conditional independence
conditions, the probability limits of both 2SLS and logit-based IV
estimators consist of the sum of complier, always-taker, never-taker, and
non-causal terms. Without further conditions, complier, always-taker, and
never-taker terms take the form of weighted-average treatment effects but the advantage of the
logit-based IV estimator is that the corresponding weights in the complier
term are always non-negative, which is not necessarily the case for the 2SLS
estimator. This advantage is particularly important under additional
conditions guaranteeing that the always-taker, never-taker, and non-causal
terms vanish. Under these conditions, the probability limit of the
logit-based IV estimator is represented by a convex combination of treatment
effects for compliers and the probability limit of the 2SLS estimator is
not. Thus, under these conditions, the logit-based IV estimator has a causal
interpretation and the 2SLS estimator does not.

In addition, we develop an augmented logit-based IV estimator that has a
causal interpretation under conditions that are more plausible than those
underlying causal interpretability of the logit-based IV and 2SLS
estimators. This estimator is similar to the logit-based IV estimator itself
but contains an extra term in the logit regression used in the first step of
the logit-based IV estimator. This term in turn originates from the binary
regression model of the treatment variable on controls using a subsample of
the data corresponding to an ex ante fixed value of the instrument.

Moreover, we construct a Hausman specification test that can be used to
check some of the conditions underlying causal interpretability of the logit-based IV estimator. The
test is based on the comparison of the logit-based IV and augmented
logit-based IV estimators and is easy to perform.

We demonstrate the usefulness of our estimators by applying them to three well-established empirical studies: \citet{AE98}, \citet{ABBKK02}, and \citet{DH20}. We find that our logit-based IV and augmented logit-based IV estimates are nearly identical to the 2SLS estimates reported in \citet{AE98} and \citet{ABBKK02}, but differ substantially from the 2SLS estimates in \citet{DH20}. Our results indicate that one of the reasons for this discrepancy in the latter case is that the 2SLS estimator may be assigning negative weights to a significant fraction of compliers. Our estimators avoid this issue and are therefore preferable in settings like that in \citet{DH20}.

Our paper contributes to the large literature discussing causal
interpretability of various parametric IV estimators in the case of
heterogeneous treatment effects. We therefore provide here only a few key
references that are particularly relevant for our work. The literature has
been started by \cite{IA94}, who gave the local average treatment effect
interpretation of the instrumental variable estimator in the case of a
binary treatment and a binary instrument without controls allowing for
heterogeneous treatment effects under the monotonicity (no defiers)
assumption. \cite{AI95} showed that in a model with saturated controls, the
2SLS estimator that includes all interactions between the instrument and
controls in the first step converges in probability to a weighted average of
control-specific local average treatment effects. \cite{A03}, \cite{K13},
and \cite{S20} obtained a similar result for the same and other parametric
instrumental variable estimators without saturated controls assuming that
the conditional mean function of the instrument given controls is linear. 
\cite{BBMT22} demonstrated that parametric instrumental variable estimators
generally lack a causal interpretation if this conditional mean function is
not linear.

The rest of the paper is organized as follows. In the next section, we
introduce the logit-based IV estimator and discuss its causal
interpretation. In Section \ref{sec: augmented estimator}, we discuss the
augmented logit-based IV estimator and compare its causal interpretability
with that of the logit-based IV estimator itself. In Section \ref{sec:
asymptotic distribution}, we derive asymptotic normality results for both
estimators. In Section \ref{sec: hausman test}, we develop a Hausman test
that can be used to check causal interpretability of the logit-based IV
estimator. In Section \ref{sec: empirical applications}, we present our empirical applications. In the Appendix, we provide all the proofs.

\section{Logit-Based IV Estimator}

\label{sec: estimator}

In this section, we propose a logit-based IV estimator and explain its
advantages over the 2SLS estimator in the potential outcome model with a
binary treatment and a binary instrument in the presence of controls. In
particular, we derive a set of conditions under which our logit-based IV
estimator has a causal interpretation and the 2SLS estimator does not.

Consider the potential outcome model with a binary treatment $T\in \{0,1\}$
and a binary instrument $Z\in \{0,1\}$: 
\begin{equation}
Y=Y(1)T+Y(0)(1-T)\text{ \  \  \  \ and \  \  \  \ }T=T(1)Z+T(0)(1-Z),  \label{1}
\end{equation}
where $Y(0),Y(1)\in \mathbb{R}$ are potential outcome values and $%
T(0),T(1)\in \{0,1\}$ are potential treatment values. In addition, let $X\in%
\mathcal{X }\subset \mathbb{R}^p$ be a vector of controls, including the constant one. We will assume
that the instrument $Z$ is independent of potential values $%
(Y(0),Y(1),T(0),T(1))$ conditional on $X$, which is a standard assumption in
the program evaluation literature, e.g. see Chapter 4.5.2 in \cite{AP09}:

\begin{assumption}
\label{as: conditional independence} $Z\perp (Y(0),Y(1),T(0),T(1)) \mid X$.
\end{assumption}

In this model, the group variable $G = (T(0),T(1))$ can take four values, $%
(0,0)$, $(0,1)$, $(1,0)$, and $(1,1)$, that are typically thought to
correspond to sub-populations of never-takers (NT), compliers (CP), defiers
(DF), and always-takers (AT), respectively, and it is customary to use
letter-based values instead of digit-based values. We will follow this
tradition and will write, for example, $G=CP$ instead of $G=(0,1)$.

For each sub-population $g\in \mathcal{G }= \{NT, CP, AT, AT\}$ and each $%
x\in \mathcal{X}$, define the $x$-conditional average treatment effect 
\begin{equation*}
\Delta_g(x) = \mathbb{E}[Y(1)-Y(0)|G=g, X=x].
\end{equation*}
For any estimator $\widehat \beta$, we will say that it has a {\em causal}
interpretation if its probability limit takes the form of a weighted-average
treatment effect 
\begin{equation}  \label{eq: causal interpretation}
\sum_{g\in \mathcal{G}}\mathbb{E}[\Delta_g(X)w_g(X)],
\end{equation}
where the weights $w_g(x)$ are non-negative for all $(g,x)\in \mathcal{G}%
\times \mathcal{X}$ and integrate to one: $\sum_{g\in \mathcal{G}}\mathbb{E}%
[w_g(X)] = 1$. In other words, the estimator has a causal interpretation if
its probability limit can be represented as a convex combination of
treatment effects. We will say that an estimator has a {\em partially causal}
interpretation if takes the form of a weighted-average treatment effect %
\eqref{eq: causal interpretation} with non-negative weights that do not
necessarily integrate to one.

In addition, we will impose the monotonicity condition as in \cite{IA94}:

\begin{assumption}
\label{as: monotonicity} $\mathbb{P}(T(1)\geq T(0)) = 1$.
\end{assumption}

This assumption excludes the sub-population of defiers. In the model without
controls $X$, \cite{IA94} used this assumption to identify the LATE, the
local average treatment effect for compliers: 
\begin{equation*}
\mathrm{LATE} = \mathbb{E}[Y(1) - Y(0) | G = CP],
\end{equation*}
which is often a quantity of interest. \cite{F07} extended this result and
showed that the LATE is identified in the model with controls as well, as
long as we impose Assumption \ref{as: conditional independence} in addition
to Assumption \ref{as: monotonicity}. \cite{F07}, as well as following
papers, e.g. \cite{BCFH17}, developed nonparametric and machine learning
estimators of the LATE in the model with controls.

In practice, however, empirical researchers often prefer simple parametric
alternatives, such as the 2SLS estimator. As argued in \cite{BBMT22}, this could be problematic. Indeed, not only may the 2SLS estimator
not converge in probability to the LATE, it may not have a causal
interpretation at all. In particular, it may not take the form of a
weighted-average treatment effect and even if it does, the weights may take
negative values and may not integrate to one.

To cope with some of these problems, we propose a simple alternative to the
2SLS estimator. Let $(Y_1,T_1,X_1,Z_1),\dots,(Y_n,T_n,X_n,Z_n)$ be a random
sample from the distribution of $(Y,T,X,Z)$. Also, let 
\begin{equation*}
\Lambda(t) = \frac{\exp(t)}{1 + \exp(t)},\quad t\in \mathbb{R }
\end{equation*}
be the logit function. Our estimator, which we refer to as \emph{the
logit-based IV estimator}, takes the following form.

\begin{algorithm}[Logit-Based IV Estimator]
\label{alg: logit based} Proceed in two steps:

\begin{enumerate}
\item Run the logit estimator of $Z$ on $X$, 
\begin{equation*}
\widehat \theta = \arg \max_{\theta \in \mathbb{R}^p} \sum_{i=1}^n \Big(Z_i
X_i^{\top}\theta - \log(1 + \exp(X_i^{\top}\theta))\Big);
\end{equation*}

\item Compute the logit-based IV estimator as the following ratio: 
\begin{equation*}
\widehat \beta_{\Lambda} = \frac{\sum_{i=1}^n Y_i (Z_i -
\Lambda(X_i^{\top}\widehat \theta))}{\sum_{i=1}^n T_i (Z_i -
\Lambda(X_i^{\top}\widehat \theta))}.
\end{equation*}
\end{enumerate}
\end{algorithm}

Note that this estimator differs from the 2SLS estimator by using the logit
regression residuals $Z_i - \Lambda(X_i^{\top}\widehat \theta)$ in the
second step instead of the linear regression residuals used by the 2SLS
estimator. For clarity of presentation, and since this will be helpful for
our discussion below, we provide the formal algorithm for the 2SLS estimator
as well.

\begin{algorithm}[2SLS Estimator]
\label{alg: 2sls} Proceed in two steps:

\begin{enumerate}
\item Run the OLS estimator of $Z$ on $X$, 
\begin{equation*}
\widehat{\gamma }=\arg \min_{\gamma \in \mathbb{R}^{p}}%
\sum_{i=1}^{n}(Z_{i}-X_{i}^{\top }\gamma )^{2};
\end{equation*}%
%
%
%

\item Compute the 2SLS estimator as the following ratio: 
\begin{equation}
\widehat{\beta }_{2SLS}=\frac{\sum_{i=1}^{n}Y_{i}(Z_{i}-X_{i}^{\top }%
\widehat{\gamma })}{\sum_{i=1}^{n}T_{i}(Z_{i}-X_{i}^{\top }\widehat{\gamma })%
}.  \label{eq: alt def 2sls}
\end{equation}
\end{enumerate}
\end{algorithm}

Algorithm \ref{alg: 2sls} may not be the usual way to define the 2SLS
estimator but it is easy to verify that it does define the 2SLS estimator by
applying the Frisch-Waugh-Lovell theorem.\footnote{Notably, our
logit-based IV estimator is not related to the \textquotedblleft forbidden
regression\textquotedblright \ discussed in Chapter 4.6.1 of \cite{AP09},
which refers to the use of the fitted value $\hat{T}_{i}$\ in the second stage OLS regression of $Y_{i}$ on $\hat T_{i}$ and $X_{i}$ obtained from a non-linear regression of $T_{i}$ on $Z_{i}$ and $X_{i}$.
Instead, our approach replaces the linear predictor $X_{i}^{\top }\widehat{%
\gamma }$\ of $Z_{i}$\ by a nonlinear predictor $\Lambda (X_{i}^{\top }%
\widehat{\theta })$, so we are \textquotedblleft partialling
out\textquotedblright \ $X_{i}$\ using a nonlinear model. Importantly, under the assumption of constant treatment effects, so that $Y(1) - Y(0) = \varrho$, our logit-based IV estimator is consistent for $\varrho$ under the same conditions as those required for consistency of the 2SLS estimator, as opposed to the forbidden regression, which requires extra conditions. In particular, it is easy to check that the logit-based IV estimator is consistent for $\varrho$ under Assumption \ref{as: conditional independence} as long as the function $x\mapsto \mathbb E[Y(0)|X=x]$ is linear.\label{fn: forbidden regression}}

Without further assumptions, neither logit-based IV nor 2SLS estimators may
have a causal interpretation. To see why this is so, we first need to
introduce some additional notation. For all $x\in \mathcal{X}$, let 
\begin{equation*}
\omega_{CP}(x) = \mathbb{P}(G = CP|X=x),\  \  \omega_{AT}(x) = \mathbb{P}%
(G=AT|X=x),\  \  \omega_{NT}(x) = \mathbb{P}(G=NT|X=x)
\end{equation*}
denote the $x$-conditional fractions of compliers, always-takers, and
never-takers in the population, respectively. Also, let 
\begin{equation}  \label{eq: definition theta0}
\theta_0 = \arg \max_{\theta \in \mathbb{R}^p} \mathbb{E}[ZX^{\top}\theta -
\log(1 + \exp(X^{\top}\theta))]
\end{equation}
and 
\begin{equation}  \label{eq: definition gamma0}
\gamma_0 = \arg \min_{\gamma \in \mathbb{R}^p} \mathbb{E}[(Z -
X^{\top}\gamma)^2]
\end{equation}
be the probability limits of the estimators $\widehat \theta$ and $\widehat
\gamma$ appearing in Algorithms \ref{alg: logit based} and \ref{alg: 2sls}.
Moreover, for all $x\in \mathcal{X}$, let 
\begin{equation*}
h_{\Lambda}(x) = \Lambda(x^{\top}\theta_0) \quad \text{and}\quad h_{2SLS}(x)
= x^{\top}\gamma_0.
\end{equation*}
The following theorem derives the probability limits of both $\widehat
\beta_{\Lambda}$ and $\widehat \beta_{2SLS}$.

\begin{theorem}
\label{thm: probability limits} Suppose that Assumptions \ref{as:
conditional independence} and \ref{as: monotonicity} are satisfied. Then
under appropriate regularity conditions,\footnote{%
To avoid distractions, we provide the list of regularity conditions for this
theorem and all other results in the Appendix.} for any $s\in[0,1]$ and $%
a\in \{ \Lambda,2SLS\}$, we have that $\widehat \beta_{a}\to_p \beta_a$,
where 
\begin{align}
\beta_a & =\frac{\mathbb{E}[\Delta_{CP}(X)\omega_{CP}(X)(s\mathbb{E}[Z|X] +
(1-s)h_a(X) - h_a(X)\mathbb{E}[Z|X])]}{\mathbb{E}[T(Z - h_a(X))]}
\label{eq: complier term} \\
& \quad+\frac{s\mathbb{E}[\Delta_{AT}(X)\omega_{AT}(X)(\mathbb{E}%
[Z|X]-h_a(X))]}{\mathbb{E}[T(Z - h_a(X))]}  \label{eq: always taker term} \\
& \quad+\frac{(1-s)\mathbb{E}[\Delta_{NT}(X)\omega_{NT}(X)(h_a(X) - \mathbb{E}%
[Z|X])]}{\mathbb{E}[T(Z - h_a(X))]}  \label{eq: never taker term} \\
& \quad+\frac{\mathbb{E}[\mathbb{E}[sY(0) + (1-s)Y(1)|X](\mathbb{E}%
[Z|X]-h_a(X))]}{\mathbb{E}[T(Z - h_a(X))]}  \label{4}
\end{align}
and
\begin{align}
& \mathbb{E}[T(Z - h_a(X))] 
 = \mathbb{E}[\omega_{CP}(X)(s\mathbb{E}[Z|X] +
(1-s)h_a(X) - h_a(X)\mathbb{E}[Z|X])]  \notag\\
& \qquad\qquad + s\mathbb{E}[\omega_{AT}(X)(\mathbb{E}%
[Z|X]-h_a(X))] + (1-s)\mathbb{E}[\omega_{NT}(X)(h_a(X) - \mathbb{E}%
[Z|X])].\label{eq: denominator general formula}
\end{align}
\end{theorem}

The proof of this theorem, as well as of all other results in the main text,
is rather simple and is provided in the Appendix. In fact, the expression
for the probability limit of the 2SLS estimator $\widehat \beta_{2SLS}$ in
this theorem is closely related to that in Proposition 1 of \cite{BBMT22}.

Theorem \ref{thm: probability limits} shows that without further
assumptions, the probability limits of both logit-based IV and 2SLS
estimators take the form of a sum of complier, always-taker, never-taker,
and non-causal terms, appearing in expressions \eqref{eq: complier term}, %
\eqref{eq: always taker term}, \eqref{eq: never taker term}, and \eqref{4},
respectively. The last term is referred to here as a non-causal term because
it does not take the form of a functional of the treatment effect $Y(1)-Y(0)$%
. In particular, this term depends non-trivially on the level of potential
outcomes $Y(0)$ and $Y(1)$. The representation for the probability limits
given in Theorem \ref{thm: probability limits} is not unique as different
values of the parameter $s$ give different representations. For example, it
is always possible to get rid of the never-taker term by substituting $s=1$
and it is always possible to get rid of the always-taker term by
substituting $s = 0$. However, without further assumptions, it is not
possible to get rid of both always-taker and never-taker terms at the same
time.

Theorem \ref{thm: probability limits} identifies at least two problems
with both logit-based IV and 2SLS estimators. First, the presence of the
non-causal term means that neither the logit-based IV nor the 2SLS estimator in
general converges in probability to a weighted-average treatment effect.
Second, the always-taker term in \eqref{eq: always taker term} shows that
the $x$-conditional average treatment effect for always-takers $%
\Delta_{AT}(x)$ has the weight 
$$
\frac{\omega_{AT}(x)(\mathbb{E}[Z|X=x] - h_a(x))}{\mathbb{E}[T(Z - h_a(x))]},
$$
which may be negative for both logit-based IV
and 2SLS estimators, and the same applies to the never-taker term in %
\eqref{eq: never taker term}. In addition, although the weights in the causal terms \eqref{eq: complier term}, \eqref{eq: always taker term}, and \eqref{eq: never taker term} do integrate to one, as we show below (Corollary \ref{cor: less simplifications}), eliminating the non-causal term in \eqref{4} may also cause parts of the causal terms to drop out. In such cases, the resulting probability limit takes the form of a weighted-average treatment effect, but the weights no longer integrate to one. These are all the problems discussed in  \cite{BBMT22} in the case of the 2SLS estimator.

Theorem \ref{thm: probability limits} also identifies the key advantage of
the logit-based IV estimator in comparison with the 2SLS estimator: for the
former, all the weights in the complier term are non-negative and
this is not necessarily the case for the latter. 
To see this advantage of the logit-based IV estimator more clearly, we now impose two additional assumptions that ensure  that the always-taker, never-taker, and non-causal terms drop out and the weights still integrate to one.

\begin{assumption}
\label{as: linear means} For some vector $\eta_0\in \mathbb{R}^p$, we have $%
\mathbb{E}[(Y(1) - Y(0))T(0) + Y(0)|X] = X^{\top}\eta_0$ with probability
one.
\end{assumption}

As we will see below, this assumption ensures that the always-taker, never-taker, and non-causal terms drop out from the general formulas in Theorem \ref{thm: probability limits}.
Although non-standard, this assumption seems rather attractive. Indeed, it
specifies a linear regression model for the conditional mean of $(Y(1) -
Y(0))T(0) + Y(0)$ given $X$, and linear regression models have a long
tradition in econometrics. In empirical work, even if regression functions
are not believed to be exactly linear, they are believed to be approximately
linear. In this sense, Assumption \ref{as: linear means} is in line with a
traditional regression analysis in economics. Moreover, as long as the
conditional mean function $x\mapsto \mathbb{E}[(Y(1) - Y(0))T(0) + Y(0)|X=x]$
is continuous, it can be well approximated by a linear combination of, say,
polynomial transformations of $x$. In such a case, Assumption \ref{as:
linear means} can be made more plausible if we replace $X$ by a set of
polynomial, or other technical, transformations of $X$. In practice, this
amounts to replacing all the $X_i$ by, say, the $q(X_i)$, where $q(\cdot) =
(q_1(\cdot),\dots,q_k(\cdot))^{\top}$ is a vector of corresponding
transformations. Moreover, under Assumption \ref{as: conditional
independence}, 
\begin{align*}
\mathbb{E}[(Y(1) - Y(0))T(0)+Y(0)|X] & = \mathbb{E}[(Y(1) -
Y(0))T(0)+Y(0)|X, Z = 0] \\
& = \mathbb{E}[Y|X,Z=0],
\end{align*}
which implies that Assumption \ref{as: linear means} is testable; se Remark \ref{rem: testing} below for details.


\begin{assumption}
\label{as: linear first stage} For some vector $\psi_0\in \mathbb{R}^p$, we
have $\mathbb{E}[T(0)|X]=X^{\top}\psi_0$ with probability one.
\end{assumption}

This assumption ensures that the weights in Theorem \ref{thm: probability limits} integrate to one despite the fact that some terms in the theorem drop out because of Assumption \ref{as: linear means}. It specifies a linear regression model for the conditional mean
of $T(0)$ given $X$. Given that $T(0)$ is a binary random variable, the
conditional mean $\mathbb{E}[T(0)|X]$ takes values in the $(0,1)$ interval,
and so this assumption may be less plausible than Assumption \ref{as: linear
means}. However, we will use this assumption mainly to make the comparison
between the logit-based IV and 2SLS estimators particularly transparent.
Without this assumption, our logit-based IV estimator still has a
partially causal interpretation as we will explain in Corollary \ref{cor: less simplifications} below. Also, under Assumption \ref{as: conditional
independence}, 
\begin{equation*}
\mathbb{E}[T(0)|X] = \mathbb{E}[T(0)|X,Z=0] = \mathbb{E}[T|X,Z=0],
\end{equation*}
which implies that Assumption \ref{as: linear first stage} is testable as
well. In addition, like Assumption \ref{as: linear means}, it can be made
more plausible if we replace $X$ by a set of appropriate transformations of $%
X$. Finally, we will discuss in the next section how one can modify the
logit-based IV estimator to accommodate more plausible versions of
Assumption \ref{as: linear first stage}.

By combining Theorem \ref{thm: probability limits} with Assumptions \ref{as:
linear means} and \ref{as: linear first stage}, we obtain the following
corollary.

\begin{corollary}
\label{cor: simplifications} Suppose that Assumptions \ref{as: conditional
independence}, \ref{as: monotonicity}, \ref{as: linear means}, and \ref{as:
linear first stage} are satisfied. Then under appropriate regularity
conditions, the probability limits $\beta_{\Lambda}$ and $\beta_{2SLS}$
appearing in Theorem \ref{thm: probability limits} take the following form: 
\begin{equation*}
\beta_{\Lambda} = \mathbb{E}[\Delta_{CP}(X)w_{\Lambda}(X)]\quad \text{and}%
\quad \beta_{2SLS} = \mathbb{E}[\Delta_{CP}(X)w_{2SLS}(X)],
\end{equation*}
where 
\begin{equation*}
w_{a}(x) = \frac{\omega_{CP}(x)\mathbb{E}[Z|X=x](1 -
h_a(x))}{\mathbb{E}[\omega_{CP}(X)\mathbb{E}[Z|X](1 -
h_a(X))]}
\end{equation*}
for all $x\in \mathcal{X}$ and $a\in\{\Lambda,2SLS\}$.
\end{corollary}

This corollary provides a clean comparison between logit-based IV and 2SLS
estimators. It shows that under our assumptions, both estimators converge in
probability to weighted-average treatment effects for compliers, and the
weights do integrate to one: 
\begin{equation*}
\mathbb{E}[w_{\Lambda}(X)] = \mathbb{E}[w_{2SLS}(X)]=1.
\end{equation*}
However, in the case of the 2SLS estimator, some of the weights may take
negative values, which happens whenever $h_{2SLS}(X) = X^{\top}\gamma_0$ exceeds one with positive probability. At the same time, the weights of the logit-based IV estimator are always
non-negative, as the logit function $\Lambda(\cdot)$ takes values in the $%
(0,1)$ interval. Thus, under our assumptions, the logit-based IV estimator
has a causal interpretation and the 2SLS estimator does not. This explains
the main advantage of the logit-based IV estimator relative to the
traditionally used 2SLS estimator.

We now discuss two extensions of Corollary \ref{cor: simplifications}.
First, without imposing Assumption \ref{as: linear first stage}, we obtain
somewhat more convoluted expressions for the probability limits of the
logit-based IV and 2SLS estimators, which are nonetheless useful to make
comparisons between these two estimators.

\begin{corollary}
\label{cor: less simplifications} Suppose that Assumptions \ref{as:
conditional independence}, \ref{as: monotonicity}, and \ref{as: linear means}
are satisfied. Then under appropriate regularity conditions, the probability
limits $\beta_{\Lambda}$ and $\beta_{2SLS}$ appearing in Theorem \ref{thm:
probability limits} take the following form: 
\begin{equation*}
\beta_{\Lambda} = \mathbb{E}[\Delta_{CP}(X)w_{\Lambda}(X)]\quad \text{and}%
\quad \beta_{2SLS} = \mathbb{E}[\Delta_{CP}(X)w_{2SLS}(X)],
\end{equation*}
where 
\begin{equation}  \label{eq: logit weights do not sum to one}
w_{a}(x) = \frac{\omega_{CP}(x)\mathbb{E}[Z|X=x](1 -
h_a(x))}{\mathbb{E}[\omega_{CP}(X)\mathbb{E}[Z|X](1 -
h_a(X))] + \mathbb{E}[\omega_{AT}(X)(\mathbb{E}[Z|X] -
h_a(X))]}
\end{equation}
for all $x\in \mathcal{X}$ and $a\in\{\Lambda,2SLS\}$.
\end{corollary}

This corollary shows that without imposing Assumption \ref{as: linear first
stage}, both logit-based IV and 2SLS estimators still converge in
probability to weighted-average treatment effects for compliers but the
weights now do not integrate to one. On the other hand, as long as the
fraction of always-takers is not too large relative to the fraction of
compliers, so that the denominators in 
\eqref{eq: logit weights do not sum to one} remain non-negative
for all $x\in \mathcal{X}$, the main advantage of the logit-based IV
estimator remains valid: its corresponding weights are still non-negative
and the estimator has a partially causal interpretation, whereas the weights
of the 2SLS estimator may be negative and the estimator does not have a
partially causal interpretation.

It is also worth noting that, for both the logit-based IV and the 2SLS estimators, the weights may sum to more or less than one. Indeed, inspection of the formula in \eqref{eq: logit weights do not sum to one} shows that, for both $a=\Lambda$ and $a=2SLS$, we would have $\mathbb{E}[w_a(X)] = 1$ if it were the case that 
$$
\mathbb{E}[\omega_{AT}(X)(\mathbb{E}[Z|X] - h_a(X))] = 0.
$$
However, in both cases, 
$$
\mathbb{E}[\mathbb{E}[Z|X] - h_a(X)] = \mathbb E[Z - h_a(X)] = 0,
$$
and so the term $\mathbb{E}[\omega_{AT}(X)(\mathbb{E}[Z|X] - h_a(X))]$ may be either positive or negative. It follows that the failure of the weights to integrate to one can generate both attenuation and amplification biases.

Second, we can relax Assumptions \ref{as: linear means} and \ref{as: linear
first stage} without losing causal interpretability of the logit-based IV
estimator. Indeed, consider the following assumption.

\begin{assumption}
\label{as: relaxed conditions} For some constant $s\in[0,1]$ and some
vectors $\eta_0$ and $\psi_0$ in $\mathbb{R}^p$, we have 
\begin{equation}  \label{eq: relaxed linear means}
\mathbb{E}\Big[(Y(1) - Y(0))(sT(0)+(1-s)T(1))+Y(0)|X\Big] = X^{\top}\eta_0
\end{equation}
and 
\begin{equation}  \label{eq: relaxed linear first stage}
\mathbb{E}[sT(0) + (1-s)T(1)|X] = X^{\top}\psi_0
\end{equation}
with probability one.
\end{assumption}

This assumption relaxes Assumptions \ref{as: linear means} and \ref{as:
linear first stage} as it reduces to Assumptions \ref{as: linear means} and %
\ref{as: linear first stage} if we plugin $s=1$. It requires that there
exists a linear combination of the conditional mean functions $x\mapsto 
\mathbb{E}[(Y(1)-Y(0))T(1)+Y(0)|X=x]$ and $x\mapsto \mathbb{E}%
[(Y(1)-Y(0))T(0)+Y(0)|X=x]$ that is linear and that the linear combination
of the conditional mean functions $x\mapsto \mathbb{E}[T(1)|X=x]$ and $%
x\mapsto \mathbb{E}[T(0)|X=x]$ with the same weights is linear as well. All
the comments we made about Assumptions \ref{as: linear means} and \ref{as:
linear first stage} apply to this assumption as well. In particular, one can
show that equations \eqref{eq: relaxed linear means} and 
\eqref{eq: relaxed
linear first stage} can be equivalently rewritten as 
\begin{equation*}
s\mathbb{E}[Y|X,Z=0] + (1-s)\mathbb{E}[Y|X,Z=1] = X^{\top}\eta_0
\end{equation*}
and 
\begin{equation*}
s\mathbb{E}[T|X,Z=0] + (1-s)\mathbb{E}[T|X,Z=1] = X^{\top}\psi_0,
\end{equation*}
respectively. Thus, Assumption \ref{as: relaxed conditions} is testable.

\begin{corollary}
\label{cor: main result} Suppose that Assumptions \ref{as: conditional
independence}, \ref{as: monotonicity}, and \ref{as: relaxed conditions} are
satisfied. Then under appropriate regularity conditions, the probability
limits $\beta_{\Lambda}$ and $\beta_{2SLS}$ appearing in Theorem \ref{thm: probability limits}
take the following form: 
\begin{equation}\label{eq: non ideal but good form}
\beta_{\Lambda} = \mathbb{E}[\Delta_{CP}(X)w_{\Lambda,s}(X)]\quad\text{and}\quad\beta_{2SLS} = \mathbb{E}[\Delta_{CP}(X)w_{2SLS,s}(X)]
\end{equation}
where 
\begin{equation}
w_{a,s}(x) = \frac{\omega_{CP}(x)\Big( s\mathbb{E}[Z|X=x] +
(1-s)h_a(x) - h_a(x)\mathbb{E}[Z|X=x] %
\Big)}{\mathbb{E}\Big[\omega_{CP}(X)\Big( s\mathbb{E}[Z|X] +
(1-s)h_a(X) - h_a(X)\mathbb{E}[Z|X] %
\Big)\Big]}\label{eq: general formula for weights relaxed}
\end{equation}
for all $x\in \mathcal{X}$ and $a\in\{\Lambda,2SLS\}$.
\end{corollary}

This corollary shows that under Assumptions \ref{as: conditional
independence} and \ref{as: monotonicity}, Assumption \ref{as: relaxed
conditions} is sufficient for causal interpretability of the logit-based IV
estimator. The weights, however, take a more complicated form than those in
Corollary \ref{cor: simplifications}. In comparison with Corollary \ref{cor: simplifications}, this corollary also shows that the weights for the 2SLS estimator may be negative not only when $\mathbb P(X^\top\gamma_0 >1)>0$ but also when $\mathbb P(X^\top\gamma_0 < 0) > 0$ (the latter case is relevant when $s$ in Assumption \ref{as: relaxed conditions} is close to zero). In practice, we therefore recommend calculating the fractions of observations $i$ with $X_i^\top\widehat\gamma > 1$ and with $X_i^\top\widehat\gamma < 0$ and disregarding the 2SLS estimator in favor of the logit-based IV estimator whenever at least one fraction is non-trivial.


Finally, we demonstrate that the logit-based IV estimator
has a causal interpretation even if Assumption \ref{as: relaxed conditions}
is not satisfied as long as it consistently estimates the conditional mean
function $x\mapsto \mathbb{E}[Z|X=x]$ on the first step. Indeed, consider the following
assumption.

\begin{assumption}
\label{as: logit correct specification} The conditional mean function $%
x\mapsto \mathbb{E}[Z|X=x]$ takes the logit form, i.e. $\mathbb{E}[Z|X] =
\Lambda(X^{\top}\theta_0)$ with probability one.
\end{assumption}

Note that in general, the assumption that the conditional mean function $%
x\mapsto \mathbb{E}[Z|X=x]$ takes the logit form means that there exists
some $\theta \in \mathbb{R}^p$ such that $\mathbb{E}[Z|X] =
\Lambda(X^{\top}\theta)$ with probability one. However, given the definition
of $\theta_0$ in \eqref{eq: definition theta0}, this $\theta$ should be
equal to $\theta_0$, and so we simply assume that $\mathbb{E}[Z|X] =
\Lambda(X^{\top}\theta_0)$.

\begin{corollary}
\label{cor: true mean} Suppose that Assumptions \ref{as: conditional
independence}, \ref{as: monotonicity}, and \ref{as: logit correct
specification} are satisfied. Then under appropriate regularity conditions,
the probability limit $\beta_{\Lambda}$ appearing in Theorem \ref{thm:
probability limits} takes the following form: 
\begin{equation}\label{eq: ideal form}
\beta_{\Lambda} = \mathbb{E}[\Delta_{CP}(X)w_{0}(X)],
\end{equation}
where 
\begin{equation}  \label{eq: true weights}
w_{0}(x) = \frac{\omega_{CP}(x)\mathbb{E}[Z|X=x](1 - \mathbb{E}[Z|X=x])}{%
\mathbb{E}[\omega_{CP}(X)\mathbb{E}[Z|X](1 - \mathbb{E}[Z|X])]}
\end{equation}
for all $x\in \mathcal{X}$.
\end{corollary}

Together with Corollary \ref{cor: main result}, this corollary shows that
our logit-based IV estimator has a double robustness property, meaning that
it has a causal interpretation if at least one of two conditions holds:
either Assumption \ref{as: relaxed conditions} or Assumption \ref{as: logit
correct specification} is satisfied. Interestingly, however, the weights $%
w_{\Lambda,s}(\cdot)$ and $w_0(\cdot)$ appearing in Corollaries \ref{cor:
main result} and \ref{cor: true mean} may be different, which means that
even though we have a causal interpretation in both cases, the interpretation of the parameter we
are estimating depends on which condition is being satisfied.

\begin{remark}\label{eq: comparison between logit and LS 1}
{\normalfont When studying the 2SLS estimator, researchers often assume that
the conditional mean function $x\mapsto \mathbb{E}[Z|X=x]$ takes the linear
form, e.g. see \cite{A03}, \cite{K13}, and \cite{S20}. In particular, with
Assumptions \ref{as: conditional independence} and \ref{as: monotonicity}
being maintained, as discussed in the Introduction, under this linear form
condition, the 2SLS estimator has a causal interpretation and, in fact, its probability limit is $\beta_{2SLS} =  \mathbb{E}[\Delta_{CP}(X)w_{0}(X)]$ with weights $w_0$ given by \eqref{eq: true weights}. It is therefore
useful to compare this linear form condition with our logit form condition
in Assumption \ref{as: logit correct specification}. With saturated
controls, the linear form condition and the logit form condition are both
satisfied, and so logit-based IV and 2SLS estimators consistently estimate the same quantity and both have a causal interpretation. Without saturated controls, however, it is unlikely that
both conditions are satisfied simultaneously. In this case, given that $Z$
is a binary random variable, so that the conditional mean function $x\mapsto 
\mathbb{E}[Z|X=x]$ takes values in the $(0,1)$ interval, the logit form
condition seems more plausible than the linear form condition, and so our
logit-based IV estimator is more likely to have a causal interpretation than
the 2SLS estimator.} \qed
\end{remark}

\begin{remark}\label{eq: comparison between logit and LS 2}
{\normalfont We now explain why we focus on the logit link function in our logit-based IV estimator rather than considering a general class of link functions. For the purposes of this remark, we assume that the function $x\mapsto \mathbb E[Z|X=x]$ takes neither linear not logit form; otherwise we fall back on the previous remark. 

Observe first that by the proof of Corollary \ref{cor: main result}, the formulas for $\beta_a$ and $\mathbb E[T(Z - h_a(X))]$ in Theorem \ref{thm: probability limits} can be rewritten as
\begin{align}
\beta_{a} & =\frac{\mathbb{E}[\Delta_{CP}(X)\omega_{CP}(X)(s\mathbb{E}[Z|X] + (1-s)h_a(X) - h_a(X)\mathbb{E}[Z|X])]}{\mathbb{E}[T(Z - h_a(X))]}
\label{eq: cor 2.3 exp 4 1} \\
& \quad +\frac{\mathbb{E}[\mathbb{E}[(Y(1) - Y(0))(sT(0)+(1-s)T(1)) + Y(0)|X](%
\mathbb{E}[Z|X] - h_a(X))]}{\mathbb{E}[T(Z - h_a(X))]}.  \label{eq: cor 2.3 exp new 2}
\end{align}
and
\begin{align}
 \mathbb{E}[T(Z - h_a(X))] 
& = \mathbb{E}[\omega_{CP}(X)(s\mathbb{E}[Z|X] +
(1-s)h_a(X) - h_a(X)\mathbb{E}[Z|X])]  \label{eq: denominator general formula 2 3}\\
& \quad + \mathbb{E}[(\mathbb E[sT(0)+ (1-s)T(1)|X] + s-1)(\mathbb{E}[Z|X]-h_a(X))],\label{eq: denominator general formula 3 4}
\end{align}
where the term in \eqref{eq: cor 2.3 exp new 2} equals the sum of the terms in \eqref{eq: always taker term}, \eqref{eq: never taker term}, and \eqref{4}. Also, since $\mathbb P(\mathbb E[Y|X]\neq h_{2SLS}(X))>0$ and $\mathbb E[\mathbb E[Z|X] - h_{2SLS}(X)] = \mathbb E[Z - h_{2SLS}(X)] = 0$ by the first-order conditions for the optimization problem \eqref{eq: definition gamma0} and the fact that $X$ includes the constant one, it follows that both $\mathbb P(\mathbb E[Z|X] -h_{2SLS}(X)>0)$ and $\mathbb P(\mathbb E[Z|X] -h_{2SLS}(X)<0)$ are strictly positive. Thus, some treatment effects in \eqref{eq: always taker term} and \eqref{eq: never taker term} have negative weights unless it somehow happens that $\Delta_{AT}(X)\omega_{AT}(X) = 0$ whenever $\mathbb E[Z|X] -h_{2SLS}(X)<0$ and $\Delta_{NT}(X)\omega_{NT}(X) = 0$ whenever $\mathbb E[Z|X] -h_{2SLS}(X)>0$, which may not be plausible. Therefore, ensuring that the 2SLS estimator has a causal interpretation essentially requires that the terms in \eqref{eq: cor 2.3 exp new 2} and \eqref{eq: denominator general formula 3 4} both vanish,\footnote{The condition that the term in \eqref{eq: denominator general formula 3 4} vanishes ensures that the weights integrate to one after dropping the term in \eqref{eq: cor 2.3 exp new 2}.} which in turn essentially requires that Assumption \ref{as: relaxed conditions} holds.\footnote{The terms in \eqref{eq: cor 2.3 exp new 2} and \eqref{eq: denominator general formula 3 4} may vanish even if Assumption \ref{as: relaxed conditions} does not hold but then we would need that $\mathbb E[f(X)(Z - h_{2SLS}(X))] = 0$ for some non-linear functions $f$, which may not be plausible.} Hence, the essentially necessary (but not sufficient) condition for the 2SLS estimator to have a causal interpretation is that Assumption \ref{as: relaxed conditions} holds. On the other hand, as follows from Corollary \ref{cor: main result}, the same assumption {\em is} sufficient for the logit-based IV estimator to have a causal interpretation. In this sense, the logit-based IV estimator improves upon the 2SLS estimator. Intuitively, this happens because the optimization problems in \eqref{eq: definition theta0} and \eqref{eq: definition gamma0} have first-order conditions of the same form, namely $\mathbb E[X(Z - h_a(X))] = 0$. If we were to use a different link function, giving the resulting estimator causal interpretability would require imposing some non-linear functional forms in Assumption \ref{as: relaxed conditions}. In that case the resulting estimator would be different from the 2SLS estimator but would not improve upon it.
}
\qed
\end{remark}

\begin{remark}\label{rem: overweighting}
{\normalfont
Since the weights $w_0(x)$ in \eqref{eq: true weights} are proportional to $\textrm{Var}(Z|X=x) = \mathbb E[Z|X=x](1-\mathbb E[Z|X=x])$, it follows that under Assumption \ref{as: logit correct specification}, units with ``more balanced assignment to treatment'' (those that have $\mathbb E[Z|X]$ closer to $1/2$) are over-weighted relative to those with ``less balanced assignment to treatment.'' On the other hand, under Assumption \ref{as: relaxed conditions}, we should replace the weights $w_0(x)$ by the weights $w_{\Lambda,s}(x)$ in \eqref{eq: general formula for weights relaxed}, and the formula in \eqref{eq: general formula for weights relaxed} is more complicated. However, as long as the function $x\mapsto h_{\Lambda}(x)$ does not deviate too much from the function $x\mapsto \mathbb{E}[Z|X=x]$, the qualitative comparison remains the same: units with ``more balanced assignment to treatment'' are relatively over-weighted. This is consistent with well-documented weighting properties of the 2SLS estimator.
}
\qed
\end{remark}

\begin{remark}\label{rem: sign reversal}
{\normalfont
We emphasize that although our logit-based IV estimator has the desirable property that its probability limit can be written, under transparent conditions, as a convex combination of treatment effects, this feature may not be sufficient in some applications. In particular, when treatment effects are heterogeneous in sign, even a convex combination may be positive or negative depending on which units receive greater weight. In such settings, a single scalar summary can obscure economically meaningful sign reversals. A useful complementary approach is therefore to partition the support of $X$
and report the corresponding estimates separately across subsets of this partition. Comparing estimates across these groups can serve as an informal diagnostic tool for detecting potential sign reversals in treatment effects.
}
\qed
\end{remark}

\begin{remark}\label{rem: multiple instruments}
{\normalfont
To conclude this section, we briefly discuss the case with multiple binary instruments. In this case, different instruments generically identify different complier groups and hence different causal effects. Any attempt to aggregate them into a single “overall” IV estimand necessarily requires additional structure and delivers a parameter whose weighting depends on the chosen combination of instruments. For transparency, we view it as preferable to report instrument-specific estimates (i.e., running our logit-based IV estimator separately for each instrument), rather than a single aggregated estimate, unless one is willing to adopt additional assumptions that justify aggregation.
}
\qed
\end{remark}

\section{Augmented Logit-Based IV Estimator}

\label{sec: augmented estimator}

In this section, we replace Assumption \ref{as: linear first stage} by a
more plausible assumption and show how one can modify the logit-based IV
estimator in order to obtain an estimator that still has a causal
interpretation. Our modified estimator is similar to the original
logit-based IV estimator but includes an extra covariate in the logit
regression of $Z$ on $X$.

Let $\Phi (\cdot )$ be a function mapping $\mathbb{R}$ to $[0,1]$ and
consider the following alternative to Assumption \ref{as: linear first stage}%
.

\begin{assumption}
\label{as: nonlinear first stage} For some vector $\psi_0\in \mathbb{R}^p$,
we have $\mathbb{E}[T(0)|X] = \Phi(X^{\top}\psi_0)$ with probability one.
\end{assumption}

If we set $\Phi (t)=\min (\max (0,t),1)$ for all $t\in \mathbb{R}$, then
Assumption \ref{as: nonlinear first stage} relaxes Assumption \ref{as:
linear first stage}. Indeed, in this case, two assumptions are the same if
the support of $X^{\top }\psi _{0}$ is contained in the $[0,1]$ interval but
Assumption \ref{as: nonlinear first stage} does not actually require the
support of $X^{\top }\psi _{0}$ to be contained in the $[0,1]$ interval. We
are, however, primarily interested in the cases where the function $\Phi
(\cdot )$ is nonlinear and smooth, e.g. $\Phi (\cdot )$ takes the logit or the
probit form. In these cases, Assumption \ref{as: nonlinear first stage}
seems more plausible than Assumption \ref{as: linear first stage}, as the
function $x\mapsto \Phi (x^{\top }\psi _{0})$ is smooth and automatically
satisfies the constraint that the conditional mean of $T(0)$ given $X$ takes
values in the $(0,1)$ interval, without restricting the support of $X$. Like
Assumptions \ref{as: linear means} and \ref{as: linear first stage},
Assumption \ref{as: nonlinear first stage} is testable and can be made more
plausible if we replace $X$ by a set of appropriate transformations of $X$.

As it turns out, we can modify our logit-based IV estimator in a way so that
it has a causal interpretation even if we replace Assumption \ref{as: linear
first stage} by Assumption \ref{as: nonlinear first stage}. The
modification, which yields \emph{the augmented logit-based IV estimator}, is
explained in the algorithm below.

\begin{algorithm}[Augmented Logit-Based IV Estimator]
\label{alg: modified logit based} Proceed in three steps:

\begin{enumerate}
\item Run the following maximum likelihood estimator using the data with $Z=0$ only, 
\begin{equation*}
\widehat{\psi }=\arg \max_{\psi \in \mathbb{R}^{p}}\sum_{i=1}^{n}\mathds%
1\{Z_{i}=0\}\Big(T_{i}\log (\Phi (X_{i}^{\top }\psi ))+(1-T_{i})\log (1-\Phi
(X_{i}^{\top }\psi ))\Big)
;
\end{equation*}

\item Run the logit estimator of $Z$ on $X$ and $\Phi (X^{\top }\widehat{%
\psi })$, 
\begin{equation*}
(\widehat{\theta },\widehat{\kappa })=\arg \max_{\theta \in \mathbb{R}%
^{p},\kappa \in \mathbb{R}}\sum_{i=1}^{n}\Big(Z_{i}(X_{i}^{\top }\theta +%
\widehat{C}_{i}\kappa )-\log (1+\exp (X_{i}^{\top }\theta +\widehat{C}%
_{i}\kappa ))\Big),
\end{equation*}%
where we denoted $\widehat{C}_{i}=\Phi (X_{i}^{\top }\widehat{\psi })$ for
all $i=1,\dots ,n$;

\item Compute the augmented logit-based IV estimator as the following ratio:%
\begin{equation*}
\widehat{\beta }_{A\Lambda }=\frac{\sum_{i=1}^{n}Y_{i}(Z_{i}-\Lambda
(X_{i}^{\top }\widehat{\theta }+\widehat{C}_{i}\widehat{\kappa }))}{%
\sum_{i=1}^{n}T_{i}(Z_{i}-\Lambda (X_{i}^{\top }\widehat{\theta }+\widehat{C}%
_{i}\widehat{\kappa }))}.
\end{equation*}
\end{enumerate}
\end{algorithm}

This estimator requires that we know the link function $\Phi$ in the
single-index structure $\Phi(X^{\top}\psi_0)$ used to model the conditional
mean of $T(0)$ given $X$. However, from a practical point of view, different
link functions, such as logit or probit, often lead to similar results, and
the researcher can always check whether it is indeed the case by trying
several link functions. Also, as pointed out above, the researcher can test
whether a particular link function is consistent with the data. Finally, the
researcher can estimate the link function $\Phi$, as described in Chapter 2
of \cite{H09}, for example.

To derive the probability limit of this estimator, define 
\begin{equation*}
\overline \psi_0 = \arg \max_{\psi \in \mathbb{R}^p}\mathbb{E}\Big[\mathds %
1\{Z = 0\} \Big(T\log(\Phi(X^{\top}\psi)) + (1-T)\log(1 - \Phi(X^{\top}\psi))%
\Big)\Big].
\end{equation*}
Note that $\bar \psi_0$ equals $\psi_0$ if Assumption \ref{as: nonlinear
first stage} is satisfied, but may be different from $\psi_0$ otherwise.
Also, define 
\begin{equation}  \label{eq: augmented logit}
(\overline \theta_0,\overline \kappa_0) = \arg \max_{\theta \in \mathbb{R}%
^p, \kappa \in \mathbb{R}} \mathbb{E}[Z(X^{\top}\theta +
\Phi(X^{\top}\overline \psi_0) \kappa) - \log(\exp(X^{\top}\theta +
\Phi(X^{\top}\overline \psi_0) \kappa))].
\end{equation}
Moreover, define 
\begin{equation*}
h_{A\Lambda}(x) = \Lambda(x^{\top}\overline \theta_0 +
\Phi(x^{\top}\overline \psi_0)\overline \kappa_0)
\end{equation*}
for all $x\in \mathcal{X}$. The following theorem provides the probability
limit of the estimator $\widehat \beta_{A\Lambda}$ imposing only Assumptions %
\ref{as: conditional independence} and \ref{as: monotonicity}.

\begin{theorem}
\label{thm: augmented estimator general} Suppose that Assumptions \ref{as:
conditional independence} and \ref{as: monotonicity} are satisfied. Then
under appropriate regularity conditions, we have that $\widehat
\beta_{A\Lambda}\to_p\beta_{A\Lambda}$, where 
\begin{align*}
\beta_{A\Lambda} & =\frac{\mathbb{E}[\Delta_{CP}(X)\omega_{CP}(X)\mathbb{E}%
[Z|X](1- h_{A\Lambda}(X))]}{\mathbb{E}[T(Z - h_{A\Lambda}(X))]} \\
& \text{ \  \ }+\frac{\mathbb{E}[\Delta_{AT}(X)\omega_{AT}(X)(\mathbb{E}%
[Z|X]-h_{A\Lambda}(X))]}{\mathbb{E}[T(Z - h_{A\Lambda}(X))]} \\
& \text{ \  \ }+\frac{\mathbb{E}[\mathbb{E}[Y(0)|X](\mathbb{E}%
[Z|X]-h_{A\Lambda}(X))]}{\mathbb{E}[T(Z - h_{A\Lambda}(X))]}.
\end{align*}
\end{theorem}

This theorem shows that the probability limit of the augmented logit-based
IV estimator has the same structure as those of the logit-based IV and 2SLS
estimators. In particular, it shows that imposing only a minimal set of
conditions used in the program evaluation literature is not sufficient to
give the augmented logit-based IV estimator a causal interpretation, which
is similar to conclusions in the previous section and in \cite{BBMT22} for
the logit-based IV and 2SLS estimators, respectively. To obtain a causal
interpretation, we need to impose additional conditions. The following
corollary provides the probability limit of the estimator $\widehat{\beta }%
_{A\Lambda }$ under the condition that either Assumption \ref{as: linear
first stage} or Assumption \ref{as: nonlinear first stage} is satisfied,
along with Assumptions \ref{as: conditional independence}, \ref{as:
monotonicity}, and \ref{as: linear means} used in the previous section.

\begin{corollary}
\label{cor: augmented estimator simplifications} Suppose that Assumptions %
\ref{as: conditional independence}, \ref{as: monotonicity} and \ref{as: linear means} are satisfied. In addition, suppose that either Assumption \ref%
{as: linear first stage} or Assumption \ref{as: nonlinear first stage} is
satisfied. Then under appropriate regularity conditions, $\widehat
\beta_{A\Lambda}\to_p \beta_{A\Lambda}$, where 
\begin{equation}\label{eq: augmented good}
\beta_{A\Lambda} = \mathbb{E}[\Delta_{CP}(X)w_{A\Lambda}(X)]
\end{equation}
and 
\begin{equation}\label{eq: augmented weights good}
w_{A\Lambda}(x) = \frac{\omega_{CP}(x)\mathbb{E}[Z|X=x](1 -
\Lambda(x^{\top}\overline \theta_0 + \Phi(x^{\top}\overline \psi_0)\overline
\kappa_0))}{\mathbb{E}[\omega_{CP}(X)\mathbb{E}[Z|X](1 -
\Lambda(X^{\top}\overline \theta_0 + \Phi(X^{\top}\overline \psi_0)\overline
\kappa_0))]}
\end{equation}
for all $x\in \mathcal{X}$.
\end{corollary}

This corollary demonstrates that the augmented logit-based IV estimator has
a causal interpretation in a wider set of cases than the logit-based IV
estimator itself. In particular, both have a causal interpretation if
Assumption \ref{as: linear first stage} is satisfied, along with other
conditions, but the former has a causal interpretation even if Assumption %
\ref{as: linear first stage} is not satisfied, as long as the correct link
function $\Phi$ is being used. Interestingly, however, the weights $%
w_{A\Lambda}(\cdot)$ appearing in this theorem are generally different from
the weights $w_{\Lambda}(\cdot)$ appearing in Corollary \ref{cor:
simplifications}, which means that even when both estimators have a causal
interpretation, they generally estimate different quantities.

In principle, Corollary \ref{cor: augmented estimator simplifications} could be extended by relaxing Assumptions \ref{as: linear means} and \ref{as: nonlinear first stage} in the same way that Corollary \ref{cor: main result} extends Corollary \ref{cor: simplifications}. However, implementing such an extension would substantially complicate Step 1 of Algorithm \ref{alg: modified logit based}. In particular, it would require estimating $\overline\psi_0$ in a model of the form
\begin{equation*}
s\mathbb{E}[T|X,Z=1] + (1-s)\mathbb{E}[T|X,Z=0] = \Phi(X^{\top}\overline
\psi_0)
\end{equation*}
for some {\em unknown} $s\in [0,1]$. Developing this extension would add considerable technical complexity without yielding additional conceptual insights, and we therefore leave it for future work. Alternatively, we could construct a variant of Algorithm \ref{alg: modified logit based} that uses the subsample with $Z=1$ in the first step rather than $Z=0$. The analysis for this variant proceeds analogously to Theorem \ref{thm: augmented estimator general} and Corollary \ref{cor: augmented estimator simplifications} upon replacing $T(0)$ by $T(1)$ in Assumptions \ref{as: linear means}, \ref{as: linear first stage}, and \ref{as: nonlinear first stage}. Since this modification is conceptually identical and does not change the main conclusions, we do not present it separately.

Moreover, we could also consider a version of the augmented logit-based IV estimator in Algorithm \ref{alg: modified logit based} that treats the data with $Z=0$ and $Z=1$ symmetrically. For such a version, on the first step of Algorithm \ref{alg: modified logit based}, we run the maximum likelihood estimator separately for the data with $Z=0$ and $Z=1$ to obtain $\widehat\psi_1$ and $\widehat\psi_2$. Then on the second and third steps, instead of including just one $\widehat C_i = \Phi(X_i^\top \widehat\psi)$, we include both $\widehat C_{i,1} = \Phi(X_i^\top\widehat\psi_1)$ and $\widehat C_{i,2} = \Phi(X_i^\top\widehat\psi_2)$, with the corresponding coefficients $\widehat\kappa_1$ and $\widehat\kappa_2$. Such an estimator has a causal interpretation if Assumptions \ref{as: linear means}, \ref{as: linear first stage}, and \ref{as: nonlinear first stage} hold either as stated, with $T(0)$, or as modified above, with $T(1)$ replacing $T(0)$.

To conclude this section, we note that as in the previous section,
Assumptions \ref{as: linear means}, \ref{as: linear first stage}, and \ref%
{as: nonlinear first stage} are not needed for a causal interpretation of
the augment logit-based IV estimator if the conditional mean function $%
x\mapsto \mathbb{E}[Z|X=x]$ takes the logit form, i.e. Assumption \ref{as:
logit correct specification} is satisfied. Indeed, we have the following
result.

\begin{corollary}
\label{cor: augmented estimator correct logit} Suppose that Assumptions \ref%
{as: conditional independence}, \ref{as: monotonicity} and \ref{as: logit
correct specification} are satisfied. Then under appropriate regularity
conditions, $\widehat \beta_{A\Lambda}\to_p \beta_{A\Lambda}$, with the
probability limit $\beta_{A\Lambda}$ taking the following form: 
\begin{equation*}
\beta_{A\Lambda} = \mathbb{E}[\Delta_{CP}(X)w_0(X)],
\end{equation*}
where 
\begin{equation}  \label{eq: true weights 2}
w_{0}(x) = \frac{\omega_{CP}(x)\mathbb{E}[Z|X=x](1 - \mathbb{E}[Z|X=x])}{%
\mathbb{E}[\omega_{CP}(X)\mathbb{E}[Z|X](1 - \mathbb{E}[Z|X])]}
\end{equation}
for all $x\in \mathcal{X}$.
\end{corollary}

Together with Corollary \ref{cor: augmented estimator simplifications}, this
corollary shows that the augmented logit-based IV estimator has a triple
robustness property, meaning that it has a causal interpretation if at least
one of three conditions holds: either Assumptions \ref{as: linear means} and %
\ref{as: linear first stage} are satisfied, Assumptions \ref{as: linear
means} and \ref{as: nonlinear first stage} are satisfied, or Assumption \ref%
{as: logit correct specification} is satisfied. In the latter case, the
logit-based IV and the augmented logit-based IV estimators have the same
probability limits as the weights $w_0(\cdot)$ in \eqref{eq: true weights}
coincide with the weights $w_0(\cdot)$ in \eqref{eq: true weights 2}, i.e. $%
\widehat \beta_{\Lambda}$ and $\widehat \beta_{A\Lambda}$ consistently estimate the same
quantity.

\section{Asymptotic Distribution Theory}

\label{sec: asymptotic distribution}

In this section, we describe the asymptotic distribution of the logit-based
IV estimator $\widehat{\beta }_{\Lambda }$ and of the augmented logit-based
IV estimator $\widehat{\beta }_{A\Lambda }$.\ We do so without imposing
Assumptions \ref{as: linear means}, \ref{as: linear first stage}, \ref{as:
relaxed conditions}, and \ref{as: nonlinear first stage} and without
assuming that the conditional mean function $x\mapsto \mathbb{E}[Z|X=x]$
takes the logit form (Assumption \ref{as: logit correct specification}). We
thus allow for general misspecification, with the probability limits $\beta
_{\Lambda }$ and $\beta _{A\Lambda }$ of the estimators being given by
formulas in Theorems \ref{thm: probability limits} and \ref{thm: augmented
estimator general}, respectively. 

To describe the asymptotic distribution of $\widehat{\beta }_{\Lambda }$,
let 
\begin{equation*}
\varphi _{0}=(\mathbb{E}[\Lambda ^{\prime }(X^{\top }\theta _{0})XX^{\top
}])^{-1}\mathbb{E}[\Lambda ^{\prime }(X^{\top }\theta _{0})X(Y-T\beta
_{\Lambda })],
\end{equation*}%
which is a vector of coefficients in the weighted projection of $Y-T\beta
_{\Lambda }$ on $X$. Also, let 
\begin{equation*}
\ell _{i}^{\Lambda }=(Y_{i}-T_{i}\beta _{\Lambda }-X_{i}^{\top }\varphi
_{0})(Z_{i}-\Lambda (X_{i}^{\top }\theta _{0}))
\end{equation*}%
for all $i=1,\dots ,n$ and let $\ell ^{\Lambda }$ be defined analogously
with $(Y,T,X,Z)$ replacing $(Y_{i},T_{i},X_{i},Z_{i})$. The following
theorem derives the asymptotic distribution of $\widehat{\beta }_{\Lambda }$.

\begin{theorem}
\label{thm: asymptotic distribution logit based} Suppose that Assumptions %
\ref{as: conditional independence} and \ref{as: monotonicity} are satisfied.
Then under appropriate regularity conditions, 
\begin{equation*}
\sqrt{n}(\widehat{\beta }_{\Lambda }-\beta _{\Lambda })=\frac{%
n^{-1/2}\sum_{i=1}^{n}\ell _{i}^{\Lambda }}{\mathbb{E}[T(Z-\Lambda (X^{\top
}\theta _{0}))]}+o_{p}(1)%
\rightarrow _{d}N(0,\sigma _{\Lambda }^{2}),
\end{equation*}%
where $\sigma _{\Lambda }^{2}=\mathbb{E}[(\ell ^{\Lambda })^{2}]/(\mathbb{E}%
[T(Z-\Lambda (X^{\top }\theta _{0}))])^{2}.$
\end{theorem}

To describe the asymptotic distribution of $\widehat{\beta }_{A\Lambda }$,
let $C=\Phi (X^{\top }\overline{\psi }_{0})$ and $C_{i}=\Phi (X_{i}^{\top }%
\overline{\psi }_{0})$
for all $i=1,\dots ,n$. Also, let $W=(X^{\top },C)^{\top }$ and $%
W_{i}=(X_{i}^{\top },C_{i})^{\top }$ for all $i=1,\dots ,n$. 
In addition, let $A_{0}=\mathbb{E}[\Lambda ^{\prime }(X^{\top }\overline{%
\theta }_{0}+C\overline{\kappa }_{0})WW^{\top }]$ and
\begin{align*}
B&=\mathbb{E}\left[\mathds{1}\{Z=0\}(1-T)\left( \frac{\Phi''(X^{\top}\overline{\psi}_0)}{1-\Phi(X^{\top}\overline\psi_0)} + \frac{\Phi'(X^\top\overline\psi_0)^2}{(1-\Phi(X^\top\overline\psi_0))^2} \right)XX^\top\right]\\ 
& \quad - \mathbb{E}\left[\mathds{1}\{Z=0\}T\left( \frac{\Phi''(X^{\top}\overline{\psi}_0)}{\Phi(X^{\top}\overline\psi_0)} - \frac{\Phi'(X^\top\overline\psi_0)^2}{\Phi(X^\top\overline\psi_0)^2} \right)XX^\top\right].
\end{align*}
Moreover, let
$$
S_i = \mathds{1}\{Z_i = 0\}\left(\frac{T_i\Phi'(X_i^\top\overline\psi_0)}{\Phi(X_i^\top\overline\psi_0)} - \frac{(1-T_i)\Phi'(X_i^\top\overline\psi_0)}{1 - \Phi(X_i^\top\overline\psi_0)}\right)X_i
$$
for all $i = 1,\dots,n$. Further, let $e$ be the vector in $\mathbb{R}^{p+1}$ such that its last component is
one and all other components are zero. Further, let 
\begin{equation*}
\xi _{0}=A_0^{-1}\Big(\mathbb{E}[\Lambda ^{\prime
}(X^{\top }\overline{\theta }_{0}+C\overline{\kappa }_{0})W(Y-T\beta
_{A\Lambda })]\Big),
\end{equation*}%
which is the vector of coefficients in the weighted projection of $Y-T\beta
_{A\Lambda }$ on $W$, 
\begin{equation*}
A_{1}=\mathbb{E}\Big[\{(Z-\Lambda (X^{\top }\overline{\theta }_{0}+C%
\overline{\kappa }_{0}))e-\overline{\kappa }_{0}\Lambda ^{\prime }(X^{\top }%
\overline{\theta }_{0}+C\overline{\kappa }_{0})W\}%
\Phi^{\prime }(X^{\top }\overline{\psi }_{0})%
X^{\top }\Big],
\end{equation*}%
\begin{equation*}
A_{2}=\overline{\kappa }_{0}\mathbb{E}\Big[\Lambda
^{\prime }(X^{\top }\overline{\theta }_{0}+C\overline{\kappa }_{0})%
\Phi^{\prime}(X^{\top }\overline{\psi }_{0})%
(Y-T\beta _{A\Lambda })X^{\top }\Big].
\end{equation*}%
Finally, let 
\begin{equation*}
\ell _{i,1}^{A\Lambda }=(Y_{i}-T_{i}\beta _{A\Lambda }-W_{i}^{\top }\xi
_{0})(Z_{i}-\Lambda (X_{i}^{\top }\overline{\theta }_{0}+C_{i}\overline{%
\kappa }_{0}))\quad\text{and}\quad \ell _{i,2}^{A\Lambda }=(\xi _{o}^{\top }A_{1}+A_{2})B^{-1}S_i
\end{equation*}%
for all $i=1,\dots ,n$ and let $\ell _{1}^{A\Lambda }$ and $\ell
_{2}^{A\Lambda }$ be defined analogously with $(Y,T,X,Z)$ replacing $%
(Y_{i},T_{i},X_{i},Z_{i})$. The following theorem derives the asymptotic
distribution of $\widehat{\beta }_{A\Lambda }$.

\begin{theorem}
\label{thm: asy dist augmented logit based estimator} Suppose that
Assumptions \ref{as: conditional independence} and \ref{as: monotonicity}
are satisfied. Then under appropriate regularity conditions, 
\begin{equation*}
\sqrt n(\widehat \beta_{A\Lambda} - \beta_{A\Lambda}) = \frac{%
n^{-1/2}\sum_{i=1}^n(\ell_{i,1}^{A\Lambda} - \ell_{i,2}^{A\Lambda})}{\mathbb{%
E}[T(Z - \Lambda(X^{\top}\overline \theta_0 + C\overline \kappa_0))]} +
o_p(1) \to_d N(0,\sigma_{A\Lambda}^2),
\end{equation*}
where $\sigma_{A\Lambda}^2 = \mathbb{E}[(\ell_{1}^{A\Lambda} -
\ell_{2}^{A\Lambda})^2]/(\mathbb{E}[T(Z - \Lambda(X^{\top}\overline \theta_0
+ C\overline \kappa_0))])^2. $
\end{theorem}

In this theorem, the terms $\ell_{i,1}^{A\Lambda}$ are analogous to the
terms $\ell_{i}^{\Lambda}$ in Theorem \ref{thm: asymptotic distribution
logit based} and the terms $\ell_{i,2}^{A\Lambda}$ capture the extra noise
appearing in Step 1 of Algorithm \ref{alg: modified logit based}.

The asymptotic variances $\sigma_{\Lambda}^2$ and $\sigma_{A\Lambda}^2$
appearing in these theorems can clearly be estimated by a plugin method, and
it is standard to provide conditions under which such estimators will be
consistent. We omit more detailed discussion for the sake of paper brevity.

\begin{remark}
{\normalfont
When $\Phi(\cdot)$ takes the logit form, i.e. $\Phi(\cdot) = \Lambda(\cdot)$, the expressions for the matrix $B$ and the vectors $S_i$ can be simplified using the identities $\Lambda'(\cdot) = \Lambda(\cdot)(1-\Lambda(\cdot))$ and $\Lambda''(\cdot) = \Lambda(\cdot)(1-\Lambda(\cdot))^2 - \Lambda(\cdot)^2(1-\Lambda(\cdot))$:
$$
B = \mathbb{E}\left[\mathds{1}\{Z=0\}\Lambda'(X^{\top}\overline{\psi}_0)XX^\top\right]
$$
and
$$
S_i = \mathds{1}\{Z_i = 0\}\left(T_i - \Lambda(X_i^\top\overline\psi_0)\right)X_i
$$
for all $i=1,\dots,n$.
}
\qed
\end{remark}

\section{Hausman Specification Test}

\label{sec: hausman test} 

In Sections \ref{sec: estimator} and \ref{sec:
augmented estimator}, we showed that one of the cases where the logit-based
IV and augmented logit-based IV estimators $\widehat \beta_{\Lambda}$ and $%
\widehat \beta_{A\Lambda}$ have a causal interpretation is the case where
the conditional mean function $x\mapsto \mathbb{E}[Z|X=x]$ takes the logit
form, i.e. Assumption \ref{as: logit correct specification} is satisfied. In
this section, we develop a Hausman test to check whether Assumption \ref{as:
logit correct specification} is indeed satisfied, following the original
work in \cite{H78}. Throughout this section, we will implicitly maintain
Assumptions \ref{as: conditional independence} and \ref{as: monotonicity}.

To describe the Hausman test, observe that under Assumption \ref{as: logit
correct specification}, it follows from Corollaries \ref{cor: true mean} and %
\ref{cor: augmented estimator correct logit} that $\beta_{\Lambda} =
\beta_{A\Lambda}$. Therefore, to test whether Assumption \ref{as: logit
correct specification} is satisfied, it makes sense to check whether the
estimators $\widehat \beta_{\Lambda}$ and $\widehat \beta_{A\Lambda}$ are
sufficiently close to each other. In turn, the asymptotic distribution of
the difference $\widehat \beta_{\Lambda} - \widehat \beta_{A\Lambda}$ can be
obtained from the asymptotic expansions in Theorems \ref{thm: asymptotic
distribution logit based} and \ref{thm: asy dist augmented logit based
estimator}. Indeed, under Assumption \ref{as: logit correct specification},
we have $\overline \theta_0 = \theta_0$ and $\overline \kappa_0 = 0$, and
so, by Theorems \ref{thm: asymptotic distribution logit based} and \ref{thm:
asy dist augmented logit based estimator}, 
\begin{equation*}
\sqrt n(\widehat \beta_{\Lambda} - \widehat \beta_{A\Lambda}) = \frac{%
n^{-1/2}\sum_{i=1}^n(\ell_i^{\Lambda} - \ell_{i,1}^{A\Lambda} +
\ell_{i,2}^{A\Lambda})}{\mathbb{E}[T(Z - \Lambda(X^{\top}\theta_0))]} +
o_p(1) \to N(0,\sigma_{H}^2),
\end{equation*}
where 
\begin{equation*}
\sigma_{H}^2 = \frac{\mathbb{E}[(\ell^{\Lambda} - \ell_{1}^{A\Lambda} +
\ell_{2}^{A\Lambda})^2]}{(\mathbb{E}[T(Z - \Lambda(X^{\top}\theta_0))])^2}.
\end{equation*}
The Hausman test therefore rejects the null hypothesis that Assumption \ref%
{as: logit correct specification} is satisfied if the test statistic 
\begin{equation}  \label{eq: hausman test statistic}
\widehat H = \frac{|\sqrt n(\widehat \beta_{\Lambda} - \widehat \beta_{A\Lambda})|}{%
\widehat \sigma_{H}}
\end{equation}
exceeds the critical value $z_{1- \alpha/2}$, where $\widehat \sigma_{H}$ is
the plugin estimator of $\sigma_{H}$, $\alpha$ is the nominal level of the
test and $z_{1- \alpha/2}$ is the number such that a standard normal random
variable exceeds this number with probability $\alpha/2$.

Being parametric, the test we have just described may not have power against
some alternatives. However, as long as Assumption \ref{as: logit correct
specification} is not satisfied, we will generically have $(\overline
\theta_0,\overline \kappa_0)\neq (\theta_0,0)$, in which case the function $%
h_{A\Lambda}(\cdot)$ is different from the function $h_{\Lambda}(\cdot)$
and, as follows from Theorems \ref{thm: probability limits} and \ref{thm:
asy dist augmented logit based estimator}, $\beta_{\Lambda}$ is different
from $\beta_{A\Lambda}$. Thus, the Hausman test will have power against most
alternatives.

In addition, we can consider a split-sample version of the Hausman test. To
describe it, randomly split the whole sample $\mathcal{I }= \{1,\dots,n\}$
into two subsamples $\mathcal{I}_1$ and $\mathcal{I}_2$ of comparable size and calculate
the logit-based IV estimator using the subsample $\mathcal{I}_1$ and the
augmented logit-based IV estimator using the subsample $\mathcal{I}_2$. Call
these estimators $\widehat \beta_{\Lambda,1}$ and $\widehat
\beta_{A\Lambda,2}$, respectively. Then, under Assumption \ref{as: logit
correct specification}, these two estimators have the same probability
limits and, by Theorems \ref{thm: asymptotic distribution logit based} and %
\ref{thm: asy dist augmented logit based estimator}, 
\begin{equation*}
\sqrt n(\widehat \beta_{\Lambda,1} - \widehat \beta_{A\Lambda,2}) = \frac{%
\sqrt{n}(|\mathcal I_1|^{-1}|\sum_{i\in \mathcal{I}_1}\ell_i^{\Lambda} - |\mathcal I_2|^{-1}|\sum_{i\in \mathcal{I}%
_2}(\ell_{i,1}^{A\Lambda} - \ell_{i,2}^{A\Lambda}))}{\mathbb{E}[T(Z -
\Lambda(X^{\top}\theta_0))]}\to_d N(0,\sigma_{H,2}^2),
\end{equation*}
where 
\begin{equation*}
\sigma_{H,2}^2 = \frac{\alpha_0 \mathbb{E}[(\ell^{\Lambda})^2] + (1-\alpha_0)\mathbb{E}%
[(\ell_{1}^{A\Lambda} - \ell_2^{A\Lambda})^2]}{(\mathbb{E}[T(Z -
\Lambda(X^{\top}\theta_0))])^2}
\end{equation*}
and $\alpha_0$ is the limit of the ratio $n/|\mathcal I_1|$. The split-sample Hausman test therefore rejects the null hypothesis that Assumption \ref{as: logit correct specification} is satisfied if the test
statistic 
\begin{equation}\label{eq: split-sample hausman test statistic}
\widehat H_S = \frac{|\sqrt n(\widehat \beta_{\Lambda,1} - \widehat \beta_{A\Lambda,2})|}{%
\widehat \sigma_{H,2}}
\end{equation}
exceeds the critical value $z_{1- \alpha/2}$, where $\widehat \sigma_{H,2}$
is the plugin estimator of $\sigma_{H,2}$ and the rest is the same as before.

Being split-sample, this version of the Hausman test may be somewhat less
powerful than the one described above. However, it is more robust in terms
of size control if it incidentally happens that $\sigma_H^2$ is close to
zero, in which case the distribution of the test statistic in 
\eqref{eq:
hausman test statistic} may not be well approximated by the standard normal
distribution.

\begin{remark}
\normalfont{The main advantage of the Hausman tests we described in this
section is their simplicity. We note, however, that there exist numerous
nonparametric tests in the literature that are more complicated to implement
but might have better power against some alternatives, e.g. see \cite{B82},
\cite{HM93}, \cite{HS01} for classical tests and \cite{J22} for recent
developments. On the other hand, it does not seem to be the case that the
power of any of these tests uniformly dominates that of the Hausman tests.} 
\qed
\end{remark}

\begin{remark}\label{rem: testing}
{\normalfont
To conclude this section, we emphasize that although the logit-based IV and augmented logit-based IV estimators improve upon the 2SLS estimator in the binary treatment and binary instrument setting as explained in Remarks \ref{eq: comparison between logit and LS 1} and \ref{eq: comparison between logit and LS 2} above, they remain parametric and thus require parametric assumptions for a causal interpretation—assumptions that can be tested on the data. Therefore, one sensible way to proceed in practice is as follows. First, run the Hausman test as described in this section to check if the conditional mean of $Z$ given $X$ has the logit form. If the test does not reject, one can conclude that both logit-based IV and augmented logit-based IV estimators have a causal interpretation and consistently estimate the weighted-average treatment effect for compliers (call it W-LATE) in \eqref{eq: ideal form} with weights given by \eqref{eq: true weights}. Otherwise, run the RESET test as described in \cite{R69} to check linearity of the conditional mean of $Y$ given $X$ and linearity of the conditional mean of $T$ given $X$ using the data with $Z=0$. If neither linearity condition is rejected, one can conclude that the logit-based IV estimator has a causal interpretation and consistently estimates the W-LATE in \eqref{eq: non ideal but good form} with weights given by \eqref{eq: general formula for weights relaxed} with $s=1$. Otherwise, run the RESET test as described in \cite{R69} to check linearity of the conditional mean of $Y$ given $X$ and linearity of the conditional mean of $T$ given $X$ using the data with $Z=1$. If neither linearity condition is rejected, one can conclude that the logit-based IV estimator has a causal interpretation and consistently estimates the W-LATE in \eqref{eq: non ideal but good form} with weights given by \eqref{eq: general formula for weights relaxed} with $s=0$. Otherwise, run the RESET test as described in \cite{P80} to check if the conditional mean of $T$ given $X$ takes the $\Phi$ form using the data with $Z=0$.\footnote{Both Ramsey and Pregibon versions of the RESET tests can be implemented in Stata.} If the test does not reject and, in addition, the linearity of $x\mapsto \mathbb E[Y|X=x,Z=0]$ was not rejected above, the augmented logit-based IV estimator has a causal interpretation and consistently estimates the W-LATE in \eqref{eq: augmented good} with weights given by \eqref{eq: augmented weights good}. Otherwise, run the RESET test as described in \cite{P80} to check if the conditional mean of $T$ given $X$ takes the $\Phi$ form using the data with $Z=1$. If the test does not reject and, in addition, the linearity of $x\mapsto \mathbb E[Y|X=x,Z=1]$ was not rejected above, the augmented logit-based IV estimator that uses the data with $Z=1$ on the first step (instead of the data with $Z=0$) has a causal interpretation and consistently estimates the W-LATE in \eqref{eq: augmented good} with weights given by \eqref{eq: augmented weights good}, with $\mathbb{E}[Z|X]$ and $\Lambda(X^{\top}\overline \theta_0 + \Phi(X^{\top}\overline \psi_0)\overline \kappa_0))$ interchanged. Otherwise, one should conclude that neither logit-based IV nor augmented logit-based IV estimator has a causal interpretation and use nonparametric estimators instead, e.g. those based on the double/debiased machine learning approach as constructed in \cite{BCFH17} and \cite{CCDDHNR18}. More broadly, we agree with conclusions in Section 7 of \cite{BBMT22} that in practice it does make sense to complement the parametric estimators, like our logit-based IV estimator, with nonparametric ones.\footnote{Our proposal involves multiple testing, and in principle the critical values of each test should be adjusted to account for multiplicity. However, we leave this issue for future research since it is outside of the main scope of the paper.}
}
\qed
\end{remark}

\section{Empirical Applications}

\label{sec: empirical applications}

In this section, we compare our logit-based IV estimators with the 2SLS estimator in three empirical applications: \citet{AE98}, \citet{ABBKK02}, and \citet{DH20}. All three applications involve a binary treatment and a binary instrument, making them well-suited for our logit-based IV estimators. 

\citet{AE98} study the effect of childbearing on women's labor supply. We focus on their 2SLS specification in Table 7, Column 2, based on a sample of 394,840 women aged 21--35 with two or more children in the 1980 Census. The outcome $Y$ represents several labor market measures: an indicator for having worked for pay in the previous year, weeks worked in the previous year, average hours worked per week, labor income, or the log of family income. The treatment $T$ is a binary indicator for having more than two children. The instrument $Z$ is a binary indicator for whether the first two children are of the same sex. The vector of controls $X$ includes age, age at first birth, plus indicators for the first child being a boy, the second child being a boy, and the mother being Black, Hispanic, or of another race.

\citet{ABBKK02} analyze the impact of the school voucher program PACES in Colombia. We consider their 2SLS specification in Table 7, Column 3, which estimates the effect of private-school scholarship use on academic and marital outcomes, using voucher status as an instrument. The sample consists of 1,147 PACES applicants from the 1995 cohort in Bogot\'{a}, with an average age of 12.6. The outcome $Y$ represents various measures: highest grade completed, grade repetitions since the lottery, or indicators for currently being in school, having finished 8th grade, or being married or living with a companion. The treatment $T$ is a binary indicator for ever having used a private-school scholarship. The instrument $Z$ is a binary indicator for winning a PACES voucher. The vector of controls $X$ includes city, year of application, phone access, age, gender, residence strata, month of interview, and type of survey.

\citet{DH20} examine the effect of queen rule on war. We focus on their 2SLS specification in Table 3, Column 3, estimated using an unbalanced panel of 3,586 observations covering 193 reigns in 18 European polities from 1480 to 1913. The outcome $Y$ indicates whether a polity is at war in a given year. The treatment $T$ is a binary indicator for whether a queen is in power. The instrument $Z$ is a binary indicator for whether the previous monarchs had a male firstborn child. The vector of controls $X$ includes indicators for whether the previous monarchs were unrelated corulers, whether they had any legitimate children (with and without missing birth years), whether the gender of the firstborn is missing, as well as polity and decade dummies.

The results for \citet{AE98}, \citet{ABBKK02}, and \citet{DH20} are presented in Tables \ref{tab:AE98}, \ref{tab:ABBKK02}, and \ref{tab:DH20}, respectively. For each table, we replicate the original 2SLS estimation results in Column 1.\footnote{\citet{DH20} use clustered standard errors but report only the $p$-values computed via the wild bootstrap. To facilitate comparison with our logit-based IV estimators, we report the plug-in clustered standard error for their 2SLS estimator.} We present logit-based IV and augmented logit-based IV estimation results in Columns 2 and 3, respectively. As described in Algorithm \ref{alg: logit based} above, in order to calculate the logit-based IV estimator, we first estimate a logit regression of the instrument on the controls to obtain the residual, defined as the instrument minus its predicted probability. We then estimate an IV regression of the outcome on the treatment, using this residual as the instrument. As described in Algorithm \ref{alg: modified logit based}, In order to calculate the augmented logit-based IV estimator, we first estimate a logit regression of the treatment on the controls---using only the subset of observations where the instrument equals zero---and obtain the predicted probability. Next, we estimate another logit regression of the instrument on the controls and the predicted probability from the first step and obtain the residual as before. Finally, we estimate an IV regression of the outcome on the treatment using the residual from the second step as the instrument. The standard errors of these estimators are computed using plug-in methods based on the asymptotic variance formulas derived in Theorems \ref{thm: asymptotic distribution logit based} and \ref{thm: asy dist augmented logit based estimator}. For \citet{DH20}, we calculate clustered standard errors using the clustering structure in the original study. The p-values for the Hausman test and for the split-sample Hausman test are given in Columns 4 and 5. In addition, each table contains information on the fraction of observations $i$ with $X_i^\top\widehat\gamma > 1$ and with $X_i^\top\widehat\gamma < 0$.

\begin{table}[t]
	\centering
	\begin{threeparttable}
	\caption{\label{tab:AE98}Comparison of Estimators: \citet{AE98}}
	\begin{tabular}{l*{5}{c}} 
		\toprule                     
		&(1)&(2)&(3)&(4)&(5) \\                     
		Dependent Variable&$\widehat{\beta}_{2SLS}$&$\widehat{\beta}_{\Lambda}$&$\widehat{\beta}_{A\Lambda}$&$\widehat{H}$& $\widehat{H}_S$ \\ 
		\midrule 
		Worked for pay        &  -0.117   &  -0.117   &  -0.117  & 0.757 & 0.250 \\                     
                              &  (0.025)  &  (0.025)  &  (0.025) &       &       \\ 
        Weeks worked          &  -5.559   &  -5.559   &  -5.531  & 0.371 & 0.273 \\                     
                              &  (1.118)  &  (1.118)  &  (1.119) &       &       \\
        Hours/week            &  -4.547   &  -4.547   &  -4.534  & 0.902 & 0.209 \\                     
                              &  (0.954)  &  (0.960)  &  (0.955) &       &       \\ 
        Labor income          &  -1903.0  &  -1902.9  &  -1902.5 & 0.979 & 0.484 \\                   
                              &  (546.4)  &  (546.7)  &  (546.7) &       &       \\ 
        $\log$(Family income) &  -0.025   &  -0.025   &  -0.024  & 0.971 & 0.641 \\                
                              &  (0.068)  &  (0.082)  &  (0.070) &       &       \\ 
        \midrule
        Fraction $X^{\top}\widehat{\gamma}>1$ & 0 &       &          &       &       \\
        Fraction $X^{\top}\widehat{\gamma}<0$ & 0 &       &          &       &       \\
        Observations          & 394,840   & 394,840   & 394,840  &       &       \\
		\bottomrule 
	\end{tabular}
	\begin{tablenotes}[flushleft]\footnotesize 	
		\item Notes: Column 1 replicates the 2SLS results in \citet{AE98}. Columns 2 and 3 employ the original and augmented logit-based IV estimators described in Algorithms \ref{alg: logit based} and \ref{alg: modified logit based}, respectively. Standard errors are in parentheses. Columns 4 and 5 report the $p$-values of the full-sample and split-sample Hausman statistics in equations \eqref{eq: hausman test statistic} and \eqref{eq: split-sample hausman test statistic}, respectively. Fractions $X^{\top}\widehat{\gamma}>1$ and $X^{\top}\widehat{\gamma}<0$ represent the proportions of observations where the 2SLS first-stage predicted values $X^{\top}\widehat{\gamma}$ exceed one and fall below zero, respectively.
	\end{tablenotes}
    \end{threeparttable}
\end{table}

\begin{table}[t]
	\centering
	\begin{threeparttable}
	\caption{\label{tab:ABBKK02}Comparison of Estimators: \citet{ABBKK02}}
	\begin{tabular}{l*{5}{c}} 
	    \toprule                     
		&(1)&(2)&(3)&(4)&(5) \\                     
		Dependent Variable&$\widehat{\beta}_{2SLS}$&$\widehat{\beta}_{\Lambda}$&$\widehat{\beta}_{A\Lambda}$ &$\widehat{H}$& $\widehat{H}_S$\\ 
		\midrule
		Highest grade completed &   0.196   &   0.197   &   0.198  & 0.929 & 0.617 \\ 
		                        &  (0.077)  &  (0.077)  &  (0.077) &       &       \\
		In school               &   0.010   &   0.010   &   0.010  & 0.928 & 0.557 \\   
		                        &  (0.030)  &  (0.030)  &  (0.030) &       &       \\ 
		Total repetitions since &  -0.100   &  -0.102   &  -0.102  & 0.929 & 0.682 \\ 
		\quad  lottery          &  (0.042)  &  (0.041)  &  (0.041) &       &       \\ 
		Finished 8th grade      &   0.151   &   0.152   &   0.152  & 0.930 & 0.546 \\ 
		                        &  (0.040)  &  (0.040)  &  (0.041) &       &       \\
		Married or living with  &  -0.013   &  -0.013   &  -0.013  & 0.933 & 0.308 \\   
		\quad companion         &  (0.009)  &  (0.009)  &  (0.009) &       &       \\
		\midrule
		Fraction $X^{\top}\widehat{\gamma}>1$&  0   &       &          &       &       \\
		Fraction $X^{\top}\widehat{\gamma}<0$&  0   &       &          &       &       \\
		Observations            &   1,147   &   1,147   &   1,147  &       &       \\
		\bottomrule 
	\end{tabular}
	\begin{tablenotes}[flushleft]\footnotesize 	
		\item Notes: Column 1 replicates the 2SLS results in \citet{ABBKK02}. Columns 2 and 3 employ the original and augmented logit-based IV estimators described in Algorithms \ref{alg: logit based} and \ref{alg: modified logit based}, respectively. Standard errors are in parentheses. Columns 4 and 5 report the $p$-values of the full-sample and split-sample Hausman statistics in equations \eqref{eq: hausman test statistic} and \eqref{eq: split-sample hausman test statistic}, respectively. Fractions $X^{\top}\widehat{\gamma}>1$ and $X^{\top}\widehat{\gamma}<0$ represent the proportions of observations where the 2SLS first-stage predicted values $X^{\top}\widehat{\gamma}$ exceed one and fall below zero, respectively.
	\end{tablenotes}
	\end{threeparttable}
\end{table}

\begin{table}[t]
	\centering
	\begin{threeparttable}
	\caption{\label{tab:DH20}Comparison of Estimators: \citet{DH20}}
	\begin{tabular}{l*{5}{c}} 
		\toprule                     
		&(1)&(2)&(3)&(4)&(5) \\                   
		Dependent Variable&$\widehat{\beta}_{2SLS}$&$\widehat{\beta}_{\Lambda}$&$\widehat{\beta}_{A\Lambda}$ &$\widehat{H}$& $\widehat{H}_S$ \\ 
		\midrule 
		In war        &   0.400   &  -0.060   &  -0.064   &  0.816  &  0.417 \\   
		              &  (0.208)  &  (0.089)  &  (0.090)  &         &        \\
		\midrule
		Fraction $X^{\top}\hat{\gamma}>1$ & 6.83\%  &&    &         &    \\
		Fraction $X^{\top}\hat{\gamma}<0$ & 10.15\% &&    &         &    \\
		Observations  &   3,586   &   3,586   &   3,586   &         &    \\ 
		\bottomrule 
	\end{tabular}
	\begin{tablenotes}[flushleft]\footnotesize 	
		\item Notes: Column 1 replicates the 2SLS results in \citet{DH20}. Columns 2 and 3 employ the original and augmented logit-based IV estimators described in Algorithms \ref{alg: logit based} and \ref{alg: modified logit based}, respectively. Clustered standard errors are reported in parentheses and are computed using the clustering structure in the original study. Columns 4 and 5 report the $p$-values of the full-sample and split-sample Hausman statistics in equations \eqref{eq: hausman test statistic} and \eqref{eq: split-sample hausman test statistic}, respectively. Fractions $X^{\top}\widehat{\gamma}>1$ and $X^{\top}\widehat{\gamma}<0$ represent the proportions of observations where the 2SLS first-stage predicted values $X^{\top}\widehat{\gamma}$ exceed one and fall below zero, respectively.
	\end{tablenotes}
	\end{threeparttable}
\end{table}

In both \citet{AE98} and \citet{ABBKK02}, the logit-based IV estimators yield estimates and standard errors that are nearly identical to those of the 2SLS estimator, which is consistent with the fact that there are no observations $i$ in these studies with $X_i^\top\widehat\gamma > 1$ or with $X_i^\top\widehat\gamma < 0$. Moreover, the p-values  for both the Hausman test and the split-sample Hausman test are well above 10\%. This evidence suggests that both the 2SLS estimator and our logit-based IV estimators are appropriate for these studies and all three have a causal interpretation.

However, in \citet{DH20}, the logit-based IV estimators yield results that differ substantially from those of the 2SLS estimator. While the 2SLS estimate is positive, the logit-based IV estimates are negative and insignificant. In this study, $6.83\%$ of the observations have $X_i^\top\widehat\gamma > 1$ and $10.15\%$ have $X_i^\top\widehat\gamma < 0$, indicating that the 2SLS estimator may be assigning negative weights to a non-trivial fraction of compliers. In this case, the 2SLS estimator lacks a causal interpretation, making the logit-based IV estimators the preferred approach. Moreover, the p-values for the Hausman tests are also well above 10\%, reinforcing the conclusion that our logit-based IV estimators are appropriate for this study.


\appendix

\section{Regularity Conditions}

\label{app: regularity conditions}

In this section, we provide regularity conditions for all the theorems and
corollaries in the paper, which were omitted in the main text.

Theorem \ref{thm: probability limits} and Corollaries \ref{cor:
simplifications}, \ref{cor: less simplifications}, \ref{cor: main result},
and \ref{cor: true mean} use the following regularity conditions:

\begin{assumption}
\label{as: regularity conditions} We have (i) $\mathbb{E}[Y^{2}]<\infty $,
(ii) $\mathbb{E}[\Vert X\Vert ^{2}]<\infty $ and 
(iii) $\mathbb{E}[\Lambda ^{\prime }(X^{\top }\theta _{0})XX^{\top }]$ is
non-singular. 
In addition, (iv) $\mathbb{E}[T(Z-\Lambda (X^{\top }\theta _{0}))]\neq 0$
when we analyze the logit-based IV estimator; and $\mathbb{E}[T(Z-X^{\top
}\gamma _{0})]\neq 0$ when we analyze the 2SLS estimator.
\end{assumption}

Theorems \ref{thm: augmented estimator general} and \ref{thm: asymptotic
distribution logit based} and Corollaries \ref{cor: augmented estimator
simplifications} and \ref{cor: augmented estimator correct logit}, uses the
following regularity conditions:

\begin{assumption}
\label{as: regularity conditions 2} We have (i) $\mathbb{E}[Y^{2}]<\infty $,
(ii) $\mathbb{E}[\Vert X\Vert ^{2}]<\infty $, (iii) $\mathbb{P}(Z=0)>0$,
(iv) $\Phi (\cdot )$ is Lipschitz-continuous, (v) $\mathbb{E}[\Lambda
^{\prime }(X^{\top }\overline{\psi }_{0})XX^{\top }|Z=0]$ is non-singular,
(vi) $\mathbb{E}[\Lambda ^{\prime }(X^{\top }\overline{\theta }_{0}+C%
\overline{\kappa }_{0})WW^{\top }]$ is non-singular for $W=(X^{\top
},C)^{\top }$ where $C=\Phi (X^{\top }\overline{\psi }_{0})$, and (vii) $\mathbb{E}[T(Z-\Lambda (X^{\top }\overline{\theta }%
_{0}+C\overline{\kappa }_{0}))]\neq 0$.
\end{assumption}

Theorem \ref{thm: asy dist augmented logit based estimator} uses the
following regularity conditions:

\begin{assumption}
\label{as: regularity conditions 3} We have (i) $\mathbb{E}[Y^{2}]<\infty $,
(ii) $\mathbb{E}[\Vert X\Vert ^{4}]<\infty $, (iii) $|\log\Phi(t)| + |\log(1-\Phi(t))| \leq c_1 + c_2t^2$ for all $t\in\mathbb R$ and some constants $c_1,c_2>0$, (iv) $|\Phi'(t)/\Phi(t)| + |\Phi'(t)/(1-\Phi(t))| \leq c_3 t$ for all $t\in\mathbb R$ and some constant $c_3>0$, (v) $|\Phi''(t)/\Phi(t)| + |\Phi''(t)/(1-\Phi(t))| \leq c_4 t$ for all $t\in\mathbb R$ and some constant $c_4>0$, (vi) $t\mapsto\Phi(t)$ is twice continuously differentiable on $\mathbb R$, (vii) $t\mapsto \log\Phi(t)$ and $t\mapsto \log(1-\Phi(t))$ are both concave on $\mathbb R$,
(viii) $A_0$ and $B$ are non-singular, and (ix) $\mathbb{E}[T(Z-\Lambda (X^{\top }%
\overline{\theta }_{0}+C\overline{\kappa }_{0}))]\neq 0$.
\end{assumption}

Note that Assumptions \ref{as: regularity conditions 3}(iii)-\ref{as: regularity conditions 3}(vii) are satisfied when $\Phi(\cdot)$ takes the logit or the probit form. 

\section{Proofs for Section \protect \ref{sec: estimator}}

\label{app: proofs section 2}

\begin{proof}[Proof of Theorem \protect \ref{thm: probability limits}]
Observe that under Assumption \ref{as: regularity conditions}, we have $%
\widehat \theta \to_p\theta_0$ and $\widehat \gamma \to_p \gamma_0$, for
example, by Theorem 2.7 in \cite{NM94}. Thus, by standard arguments, again
under Assumption \ref{as: regularity conditions}, 
\begin{equation}  \label{eq: probability limit general logit}
\frac{\sum_{i=1}^n Y_i(Z_i - \Lambda(X_i^{\top}\widehat \theta))}{%
\sum_{i=1}^n T_i(Z_i - \Lambda(X_i^{\top}\widehat \theta))} \to_p \frac{%
\mathbb{E}[Y(Z - \Lambda(X^{\top}\theta_0))]}{E[T(Z -
\Lambda(X^{\top}\theta_0))]}
\end{equation}
and 
\begin{equation}  \label{eq: probability limit general 2sls}
\frac{\sum_{i=1}^n Y_i(Z_i - X_i^{\top}\widehat \gamma)}{\sum_{i=1}^n
T_i(Z_i - X_i^{\top}\widehat \gamma)} \to_p \frac{\mathbb{E}[Y(Z -
X^{\top}\gamma_0)]}{\mathbb{E}[T(Z - X^{\top}\gamma_0)]}.
\end{equation}
We thus only need to show that the probability limits appearing here
coincide with those in the statement of Theorem \ref{thm: probability limits}%
, which is what we do below.

We start with the numerators. Fix $a\in \{ \Lambda,2SLS\}$ and $s\in[0,1]$. Let $\Delta Y = Y(1) - Y(0)$
and $\Delta T = T(1) - T(0)$. Then $Y$ can be decomposed as 
\begin{equation}  \label{eq: general limit proof 1}
Y = \Delta Y T + Y(0) = \Delta Y(\Delta T Z + T(0)) + Y(0) = \Delta Y \Delta
T Z + \Delta Y T(0) + Y(0).
\end{equation}
Here, 
\begin{align}
\mathbb{E}[\Delta Y \Delta T Z(Z-h_a(X))] & = \mathbb{E}[\Delta Y \Delta T
Z] - \mathbb{E}[\Delta Y \Delta T Z h_a(X)]  \notag \\
& = \mathbb{E}[\mathbb{E}[\Delta Y\Delta T|X]\mathbb{E}[Z|X]] - \mathbb{E}[%
\mathbb{E}[\Delta Y\Delta T|X]\mathbb{E}[Z|X]h_a(X)]  \notag \\
& = \mathbb{E}[\Delta_{CP}(X)\omega_{CP}(X)(\mathbb{E}[Z|X] - \mathbb{E}%
[Z|X]h_a(X))],  \label{eq: new theorem proof 1}
\end{align}
where the first equality follows from the fact that $Z\in \{0,1\}$, the
second from the law of iterated expectations and Assumption \ref{as:
conditional independence}, and the third from Lemma \ref{lem: simple lemma}.
In addition, 
\begin{align}
\mathbb{E}[\Delta Y T(0)(Z - h_a(X))] & = \mathbb{E}[\mathbb{E}[\Delta Y
T(0)|X]\mathbb{E}[Z - h_a(X)|X]]  \notag \\
& = \mathbb{E}[\Delta_{AT}(X)\omega_{AT}(X)(\mathbb{E}[Z|X] - h_a(X))],
\label{eq: new theorem proof 2}
\end{align}
where the first equality follows from the law of iterated expectations and
Assumption \ref{as: conditional independence} and the second from Lemma \ref%
{lem: simple lemma}. Moreover, 
\begin{equation}  \label{eq: new theorem proof 3}
\mathbb{E}[Y(0)(Z - h_a(X))] = \mathbb{E}[\mathbb{E}[Y(0)|X](\mathbb{E}[Z|X]
- h_a(X))]
\end{equation}
by the law of iterated expectations and Assumption \ref{as: conditional
independence}. Combining \eqref{eq: general limit proof 1}, 
\eqref{eq: new
theorem proof 1}, \eqref{eq: new theorem proof 2}, and 
\eqref{eq: new
theorem proof 3} gives 
\begin{align}
\mathbb{E}[Y(Z - h_a(X))] & = \mathbb{E}[\Delta_{CP}(X)\omega_{CP}(X)(%
\mathbb{E}[Z|X] - \mathbb{E}[Z|X]h_a(X))]  \notag \\
& \quad + \mathbb{E}[\Delta_{AT}(X)\omega_{AT}(X)(\mathbb{E}[Z|X] - h_a(X))]
\notag \\
& \quad + \mathbb{E}[\mathbb{E}[Y(0)|X](\mathbb{E}[Z|X] - h_a(X))].
\label{eq: new theorem proof 4}
\end{align}
Further, $Y$ can also be decomposed as 
\begin{align*}
Y & = Y(1) - \Delta Y(1 - T) = Y(1) - \Delta Y(1 - T(1) + \Delta T(1-Z)) \\
& = -\Delta Y\Delta T(1-Z) - \Delta Y(1 - T(1)) + Y(1).
\end{align*}
Here, 
\begin{align}
\mathbb{E}[-\Delta Y \Delta T (1-Z)(Z-h_a(X))] & = \mathbb{E}[\Delta Y
\Delta T h_a(X)] - \mathbb{E}[\Delta Y \Delta T Z h_a(X)]  \notag \\
& = \mathbb{E}[\mathbb{E}[\Delta Y\Delta T|X]h_a(X)] - \mathbb{E}[\mathbb{E}%
[\Delta Y\Delta T|X]\mathbb{E}[Z|X]h_a(X)]  \notag \\
& = \mathbb{E}[\Delta_{CP}(X)\omega_{CP}(X)(h_a(X) - \mathbb{E}[Z|X]h_a(X))],
\label{eq: new theorem proof 5}
\end{align}
where the first equality follows from the fact that $Z\in \{0,1\}$, the
second from the law of iterated expectations and Assumption \ref{as:
conditional independence}, and the third from Lemma \ref{lem: simple lemma}.
In addition, 
\begin{align}
\mathbb{E}[-\Delta Y(1- T(1))(Z - h_a(X))] & = \mathbb{E}[\mathbb{E}[-\Delta
Y (1-T(1))|X]\mathbb{E}[Z - h_a(X)|X]]  \notag \\
& = \mathbb{E}[\Delta_{NT}(X)\omega_{NT}(X)(\mathbb{E}[Z|X] - h_a(X))],
\label{eq: new theorem proof 6}
\end{align}
where the first equality follows from the law of iterated expectations and
Assumption \ref{as: conditional independence} and the second from Lemma \ref%
{lem: simple lemma}. Moreover, 
\begin{equation}  \label{eq: new theorem proof 7}
\mathbb{E}[Y(1)(Z - h_a(X))] = \mathbb{E}[\mathbb{E}[Y(1)|X](\mathbb{E}[Z|X]
- h_a(X))]
\end{equation}
by the law of iterated expectations and Assumption \ref{as: conditional
independence}. Combining \eqref{eq: new theorem proof 4}, 
\eqref{eq: new
theorem proof 5}, \eqref{eq: new theorem proof 6}, and 
\eqref{eq: new
theorem proof 7} gives 
\begin{align}
\mathbb{E}[Y(Z - h_a(X))] & = \mathbb{E}[\Delta_{CP}(X)\omega_{CP}(X)(h_a(X)
- \mathbb{E}[Z|X]h_a(X))]  \notag \\
& \quad - \mathbb{E}[\Delta_{NT}(X)\omega_{NT}(X)(\mathbb{E}[Z|X] - h_a(X))]
\notag \\
& \quad + \mathbb{E}[\mathbb{E}[Y(1)|X](\mathbb{E}[Z|X] - h_a(X))].
\label{eq: new theorem proof 8}
\end{align}
Writing $Y = sY + (1-s)Y$, using \eqref{eq: new theorem proof 4} and %
\eqref{eq: new theorem proof 8}, and substituting the resulting expression
into \eqref{eq: probability limit general logit} and 
\eqref{eq: probability
limit general 2sls} gives the formulas for the numerators in \eqref{eq: probability limit general logit} and \eqref{eq: probability limit general 2sls}.

Next, we consider the denominators. Again fix $a\in \{ \Lambda,2SLS\}$ and $s\in[0,1]$. Note that 
\begin{equation}  \label{eq: treatment decomposition 0}
T = \Delta T Z + T(0),
\end{equation}
where $\Delta T = T(1) - T(0)$. Also, 
\begin{align}
\mathbb{E}[\Delta T Z (Z-h_a(X))] & = \mathbb{E}[\Delta T Z] - \mathbb{E}%
[\Delta T Z h_a(X)]  \notag \\
& = \mathbb{E}[\mathbb{E}[\Delta T|X]\mathbb{E}[Z|X]] - \mathbb{E}[\mathbb{E}%
[\Delta T|X]\mathbb{E}[Z|X]h_a(X)]  \notag \\
& = \mathbb{E}[\omega_{CP}(X)(\mathbb{E}[Z|X] - h_a(X)\mathbb{E}[Z|X])] ,
\label{eq: treatment decomposition 1}
\end{align}
where the first equality follows from the fact that $Z\in \{0,1\}$, the
second from the law of iterated expectations and Assumption \ref{as:
conditional independence}, and the third from Assumption \ref{as:
monotonicity}. Also, 
\begin{equation}
\mathbb{E}[T(0)(Z-h_a(X))] = \mathbb{E}[\mathbb{E}[T(0)|X](\mathbb{E}[Z|X]
- h_a(X))] = \mathbb{E}[\omega_{AT}(X)(\mathbb{E}[Z|X] - h_a(X))]
\label{eq: treatment decomposition 2}
\end{equation}
where the first equality follows from the law of iterated expectations and
Assumption \ref{as: conditional independence} and the second from Assumption %
\ref{as: monotonicity}. Combining %
\eqref{eq: treatment decomposition 0}, \eqref{eq: treatment decomposition 1}%
, and \eqref{eq: treatment decomposition 2} gives 
\begin{equation}
\mathbb{E}[T(Z - h_a(X))] 
= \mathbb{E}[\omega_{CP}(X)(\mathbb{E}[Z|X] -
h_a(X)\mathbb{E}[Z|X])] + \mathbb{E}[\omega_{AT}(X)(\mathbb{E}[Z|X] - h_a(X))].\label{eq: denominator derivation part 1}
\end{equation}
Further,
\begin{equation}\label{eq: denominator derivation part 2}
\mathbb E[(\omega_{CP}(X) + \omega_{AT}(X) + \omega_{NT}(X))(\mathbb{E}[Z|X] - h_a(X))] = \mathbb E[\mathbb{E}[Z|X] - h_a(X)] = 0,
\end{equation}
where the first equality follows from Assumption \ref{as: monotonicity} and the second from the first order conditions in \eqref{eq: definition theta0} and \eqref{eq: definition gamma0} since $X$ includes the constant one. Multiplying \eqref{eq: denominator derivation part 2} by $s-1$ and adding the result to \eqref{eq: denominator derivation part 1} gives the formulas for the denominators in \eqref{eq: probability limit general logit} and \eqref{eq: probability limit general 2sls} and completes the proof of the corollary.
\end{proof}

\begin{proof}[Proof of Corollary \protect \ref{cor: simplifications}]
By the law of iterated expectations and first-order conditions corresponding
to the optimization problems \eqref{eq: definition theta0} and 
\eqref{eq: definition gamma0}, we have 
\begin{equation}  \label{eq: first order condition theta}
\mathbb{E}[X(\mathbb{E}[Z|X] - \Lambda(X^{\top}\theta_0))] = \mathbb{E}[X(Z
- \Lambda(X^{\top}\theta_0))] = 0
\end{equation}
and 
\begin{equation}  \label{eq: first order condition gamma}
\mathbb{E}[X(\mathbb{E}[Z|X] - X^{\top}\gamma_0)] = \mathbb{E}[X(Z -
X^{\top}\gamma_0)] = 0,
\end{equation}
respectively.

Now, fix $a\in \{ \Lambda,2SLS\}$ and apply Theorem \ref{thm: probability
limits} with $s=1$ to obtain 
\begin{align}
\beta_a & =\frac{\mathbb{E}[\Delta_{CP}(X)\omega_{CP}(X)(\mathbb{E}[Z|X] -
h_a(X)\mathbb{E}[Z|X])]}{\mathbb{E}[T(Z - h_a(X))]}
\label{eq: cor 2.1 complier term} \\
& \quad+\frac{\mathbb{E}[\Delta_{AT}(X)\omega_{AT}(X)(\mathbb{E}%
[Z|X]-h_a(X))]}{\mathbb{E}[T(Z - h_a(X))]}
\label{eq: cor 2.1 always taker term} \\
& \quad+\frac{\mathbb{E}[\mathbb{E}[Y(0)|X](\mathbb{E}[Z|X]-h_a(X))]}{%
\mathbb{E}[T(Z - h_a(X))]}.  \label{eq: cor 2.1 noncausal term}
\end{align}
Also, by Lemma \ref{lem: simple lemma}, 
\begin{equation}\label{eq: cor 21 at formula}
\Delta_{AT}(X)\omega_{AT}(X) = \mathbb{E}[(Y(1) - Y(0))T(0)|X].
\end{equation}
Hence, the sum of terms in \eqref{eq: cor 2.1 always taker term} and %
\eqref{eq: cor 2.1 noncausal term} is equal to 
\begin{equation}\label{eq: cor 21 kill the terms}
\mathbb{E}[\mathbb{E}[(Y(1) - Y(0))T(0) + Y(0)|X](\mathbb{E}[Z|X] - h_a(X))]
= \mathbb{E}[\eta_0^{\top}X(\mathbb{E}[Z|X] - h_a(X))] = 0
\end{equation}
by Assumption \ref{as: linear means}. Therefore, it remains to derive an
appropriate expression for the denominator in 
\eqref{eq: cor 2.1 complier term}. To do so, substituting $s=1$ into \eqref{eq: denominator general formula} yields
\begin{equation}\label{eq: cor 21 replacement}
\mathbb E[T(Z-h_a(X))] = \mathbb E[\omega_{CP}(X)\mathbb E[Z|X](1-h_a(X))] + \mathbb E[\omega_{AT}(X)(E[Z|X] - h_a(X))].
\end{equation}
Here, $\omega_{AT}(X) = \mathbb E[T(0)|X]$ by Assumption \ref{as: linear first stage}, and so
$$
\mathbb E[\omega_{AT}(X)(E[Z|X] - h_a(X))] = 0
$$
by \eqref{eq: first order condition theta} and \eqref{eq: first order condition gamma}. Combining this equality with \eqref{eq: cor 21 replacement} gives the denominator in \eqref{eq: cor 2.1 complier term} and completes the proof of the corollary.
\end{proof}

\begin{proof}[Proof of Corollary \protect \ref{cor: less simplifications}]
The asserted claim follows from substituting \eqref{eq: cor 21 at formula}, \eqref{eq: cor 21 kill the terms}, and \eqref{eq: cor 21 replacement} into the expression for $\beta_a$ in \eqref{eq: cor 2.1 complier term}, \eqref{eq: cor 2.1 always taker term}, and \eqref{eq: cor 2.1 noncausal term} in the proof of Corollary \ref{cor: simplifications}.
\end{proof}

\begin{proof}[Proof of Corollary \protect \ref{cor: main result}]
Fix $a\in\{\Lambda,2SLS\}$ and observe that by Lemma \ref{lem: simple lemma}, 
\begin{equation}  \label{eq: cor 2.3 exp 1}
\Delta_{AT}(X)\omega_{AT}(X) = \mathbb{E}[(Y(1) - Y(0))T(0)|X]
\end{equation}
and 
\begin{equation}  \label{eq: cor 2.3 exp 2}
\Delta_{NT}(X)\omega_{NT}(X) = \mathbb{E}[(Y(1) - Y(0))(1 - T(1))|X].
\end{equation}
Also, for any $s\in[0,1]$, 
\begin{equation}  \label{eq: cor 2.3 exp 3}
(s-1)(Y(1) - Y(0))+sY(0)+(1-s)Y(1) = Y(0).
\end{equation}
Now, let $s$ be the number in $[0,1]$ appearing in Assumption \ref{as:
relaxed conditions}. Applying Theorem \ref{thm: probability limits} with
this $s$ and using \eqref{eq: cor 2.3 exp 1}, \eqref{eq: cor 2.3 exp 2}, %
\eqref{eq: cor 2.3 exp 3} gives 
\begin{align}
\beta_{a} & =\frac{\mathbb{E}[\Delta_{CP}(X)\omega_{CP}(X)(s\mathbb{E}[Z|X] + (1-s)h_a(X) - h_a(X)\mathbb{E}[Z|X])]}{\mathbb{E}[T(Z - h_a(X))]}
\label{eq: cor 2.3 exp 4} \\
& \quad +\frac{\mathbb{E}[\mathbb{E}[(Y(1) - Y(0))(sT(0)+(1-s)T(1)) + Y(0)|X](%
\mathbb{E}[Z|X] - h_a(X))]}{\mathbb{E}[T(Z - h_a(X))]}.  \label{eq: cor 2.3 exp new}
\end{align}
However, the term in \eqref{eq: cor 2.3 exp new} is zero by Assumption \ref{as: relaxed conditions}
and \eqref{eq: first order condition theta}-\eqref{eq: first order condition gamma} in the proof of
Corollary \ref{cor: simplifications}. Therefore, it remains to derive an
appropriate expression for the denominator in \eqref{eq: cor 2.3 exp 4}. To
do so, note that $\omega_{AT}(X) = \mathbb E[T(0)|X]$ and $\omega_{NT}(X) = 1 - \omega_{CP}(X) - \omega_{AT}(X) = 1 - \mathbb E[T(1)|X]$ by Assumption \ref{as: monotonicity}. Substituting these equalities into \eqref{eq: denominator general formula} in Theorem \ref{thm: probability limits} gives
\begin{align}
 \mathbb{E}[T(Z - h_a(X))] 
& = \mathbb{E}[\omega_{CP}(X)(s\mathbb{E}[Z|X] +
(1-s)h_a(X) - h_a(X)\mathbb{E}[Z|X])]  \label{eq: denominator general formula 2}\\
& \quad + \mathbb{E}[(\mathbb E[sT(0)+ (1-s)T(1)|X] + s-1)(\mathbb{E}[Z|X]-h_a(X))].\label{eq: denominator general formula 3}
\end{align}
However, the term in \eqref{eq: denominator general formula 3} is zero by Assumption \ref{as: relaxed conditions}, \eqref{eq: first order condition theta}-\eqref{eq: first order condition gamma} in the proof of
Corollary \ref{cor: simplifications}, and the fact that $X$ includes the constant one. The asserted claim follows.
\end{proof}

\begin{proof}[Proof of Corollary \protect \ref{cor: true mean}]
The asserted claim follows from applying Theorem \ref{thm: probability limits} with any $s\in[0,1]$.
\end{proof}

\section{Proofs for Section \protect \ref{sec: augmented estimator}}

\begin{proof}[Proof of Theorem \protect \ref{thm: augmented estimator general}%
]
Observe that under Assumption \ref{as: regularity conditions 2}, we have $%
\widehat \psi \to_p\overline \psi_0$ by Theorem 2.7 in \cite{NM94}. Thus,
the function 
\begin{equation*}
(\theta, \kappa) \mapsto \widehat Q(\theta, \kappa) = \frac{1}{n}%
\sum_{i=1}^n \Big(Z_i(X_i^{\top}\theta + \widehat C_i \kappa) - \log(1 +
\exp(X_i^{\top}\theta + \widehat C_i \kappa))\Big)
\end{equation*}
converges in probability point-wise for all $(\theta, \kappa)\in \mathbb{R}%
^p\times \mathbb{R}$ to the function 
\begin{equation*}
(\theta, \kappa) \mapsto Q_0(\theta, \kappa) = \mathbb{E}[Z(X^{\top}\theta +
\Phi(X^{\top}\overline \psi_0) \kappa) - \log(\exp(X^{\top}\theta +
\Phi(X^{\top}\overline \psi_0) \kappa))].
\end{equation*}
Hence, given that both functions are concave, under Assumption \ref{as:
regularity conditions 2}, again by Theorem 2.7 in \cite{NM94}, $(\widehat
\theta,\widehat \kappa)\to_p(\overline \theta_0,\overline \kappa_0)$.
Therefore, by the standard arguments, 
\begin{equation*}
\widehat \beta_{A\Lambda} \to_p \frac{\mathbb{E}[Y(Z -
\Lambda(X^{\top}\overline \theta_0 + \Phi(X^{\top}\overline \psi_0)\overline
\kappa_0))]}{\mathbb{E}[T(Z - \Lambda(X^{\top}\overline \theta_0 +
\Phi(X^{\top}\overline \psi_0)\overline \kappa_0))]}.
\end{equation*}
The rest of the proof coincides with that in Theorem \ref{thm: probability
limits}.
\end{proof}

\begin{proof}[Proof of Corollary \protect \ref{cor: augmented estimator
simplifications}]
By the first-order conditions corresponding to the optimization problem %
\eqref{eq: augmented logit}, we have 
\begin{equation}  \label{eq: foc augmented logit 1}
\mathbb{E}[X(\mathbb{E}[Z|X] - \Lambda(X^{\top}\overline \theta_0 +
\Phi(X^{\top}\overline \psi_0)\overline \kappa_0))] = \mathbb{E}[X(Z -
\Lambda(X^{\top}\theta_0 + \Phi(X^{\top}\overline \psi_0)\overline
\kappa_0))] = 0
\end{equation}
and 
\begin{align}
& \mathbb{E}[\Phi(X^{\top}\overline \psi_0)(\mathbb{E}[Z|X] -
\Lambda(X^{\top}\overline \theta_0 + \Phi(X^{\top}\overline \psi_0)\overline
\kappa_0))]  \notag \\
& \qquad \qquad \qquad \qquad \qquad = \mathbb{E}[\Phi(X^{\top}\overline
\psi_0)(Z - \Lambda(X^{\top}\theta_0 + \Phi(X^{\top}\overline
\psi_0)\overline \kappa_0))] = 0.  \label{eq: foc augmented logit 2}
\end{align}
The proof therefore follows along the lines in the proof of Corollary \ref%
{cor: simplifications} using the probability limit in Theorem \ref{thm:
augmented estimator general} instead of the probability limit in Theorem \ref%
{thm: probability limits} and relying on the first-order condition %
\eqref{eq: foc augmented logit 1} if Assumption \ref{as: linear first stage}
holds and on \eqref{eq: foc augmented logit 2} if Assumption \ref{as:
nonlinear first stage} holds.
\end{proof}

\begin{proof}[Proof of Corollary \protect \ref{cor: augmented estimator
correct logit}]
Under Assumption \ref{as: logit correct specification}, it follows from the
optimization problem in \eqref{eq: augmented logit} that $\overline \kappa_0
= 0$ and $\overline \theta_0 = \theta_0$, so that $h_{A\Lambda}(X) =
\Lambda(X^{\top}\theta_0) = \mathbb{E}[Z|X]$ with probability one. Hence,
the asserted claim follows from Theorem \ref{thm: augmented estimator
general}.
\end{proof}

\section{Proofs for Section \protect \ref{sec: asymptotic distribution}}

\begin{proof}[Proof of Theorem \protect \ref{thm: asymptotic distribution
logit based}]
By the standard asymptotic normality result for the logit estimator, under 
Assumption \ref{as: regularity conditions}, 

\begin{equation*}
\sqrt{n}(\widehat{\theta }-\theta _{0})=\Big(\mathbb{E}[\Lambda ^{\prime
}(X^{\top }\theta _{0})XX^{\top }]\Big)^{-1}\frac{1}{\sqrt{n}}%
\sum_{i=1}^{n}(Z_{i}-\Lambda (X_{i}^{\top }\theta _{0}))X_{i}+o_{p}(1).
\end{equation*}%
Also, 
\begin{align*}
& \frac{1}{\sqrt{n}}\sum_{i=1}^{n}(Y_{i}-T_{i}\beta _{\Lambda
})(Z_{i}-\Lambda (X_{i}^{\top }\widehat{\theta }))=\frac{1}{\sqrt{n}}%
\sum_{i=1}^{n}(Y_{i}-T_{i}\beta _{\Lambda })(Z_{i}-\Lambda (X_{i}^{\top
}\theta _{0})) \\
& \qquad \qquad \qquad -\frac{1}{n}\sum_{i=1}^{n}(Y_{i}-T_{i}\beta _{\Lambda
})\Lambda ^{\prime }(X_{i}^{\top }\theta _{0})X_{i}^{\top }\sqrt{n}(\widehat{%
\theta }-\theta _{0})+o_{p}(1) \\
& \qquad \qquad =\frac{1}{\sqrt{n}}\sum_{i=1}^{n}(Y_{i}-T_{i}\beta _{\Lambda
}-X_{i}^{\top }\varphi _{0})(Z_{i}-\Lambda (X_{i}^{\top }\theta
_{0}))+o_{p}(1)=\frac{1}{\sqrt{n}}\sum_{i=1}^{n}\ell _{i}^{\Lambda
}+o_{p}(1).
\end{align*}%
In addition, 
\begin{equation*}
\frac{1}{n}\sum_{i=1}^{n}T_{i}(Z_{i}-\Lambda (X_{i}^{\top }\widehat{\theta }%
))\rightarrow _{p}\mathbb{E}[T(Z-\Lambda (X^{\top }\theta _{0}))].
\end{equation*}%
Thus, 
\begin{align*}
\sqrt{n}(\widehat{\beta }-\beta _{\Lambda })& =\frac{n^{-1/2}%
\sum_{i=1}^{n}(Y_{i}-T_{i}\beta _{\Lambda })(Z_{i}-\Lambda (X_{i}^{\top }%
\widehat{\theta }))}{n^{-1}\sum_{i=1}^{n}T_{i}(Z_{i}-\Lambda (X_{i}^{\top
}\theta _{0}))} \\
& =\frac{n^{-1/2}\sum_{i=1}^{n}\ell _{i}^{\Lambda }}{\mathbb{E}[T(Z-\Lambda
(X^{\top }\theta _{0}))]}+o_{p}(1)\rightarrow _{d}N(0,\sigma _{\Lambda
}^{2}),
\end{align*}%
yielding the asserted claim.
\end{proof}

\begin{proof}[Proof of Theorem \protect \ref{thm: asy dist augmented logit
based estimator}]
By Theorem 2.7 and the proof of Theorem 3.1, both in \cite{NM94}, under
Assumption \ref{as: regularity conditions 3}, 
\begin{equation*}
\sqrt{n}(\widehat{\psi }-\overline{\psi }_{0}) = B^{-1}\frac{1}{\sqrt n}\sum_{i=1}^n S_i + o_p(1).
\end{equation*}
Thus,
\begin{align*}
& \frac{1}{\sqrt{n}}\sum_{i=1}^{n}\Big(Z_{i}-\Lambda (X_{i}^{\top }\overline{%
\theta }_{0}+\widehat{C}_{i}\overline{\kappa }_{0})\Big)\widehat W_i \\
& \qquad =\frac{1}{\sqrt{n}}\sum_{i=1}^{n}\Big(Z_{i}-\Lambda (X_{i}^{\top }%
\overline{\theta }_{0}+C_{i}\overline{\kappa }_{0})\Big)W_i +A_{1}B^{-1}\frac{1}{\sqrt n}\sum_{i=1}^n S_i+o_{p}(1) = O_p(1),
\end{align*}
and so, by the proof of Theorem 3.1 in \cite{NM94} again,
\begin{align*}
& \sqrt{n}\left( \left( 
\begin{array}{c}
\widehat{\theta } \\ 
\widehat{\kappa }%
\end{array}%
\right) -\left( 
\begin{array}{c}
\overline{\theta }_{0} \\ 
\overline{\kappa }_{0}%
\end{array}%
\right) \right) 
 =A_{0}^{-1}\frac{1}{\sqrt{n}}\sum_{i=1}^{n}\Big(%
Z_{i}-\Lambda (X_{i}^{\top }\overline{\theta }_{0}+\widehat{C}_{i}\overline{%
\kappa }_{0})\Big) \widehat W_i +o_{p}(1) \\
& \qquad\qquad\qquad = A_0^{-1}\frac{1}{\sqrt{n}}\sum_{i=1}^{n}\Big(Z_{i}-\Lambda (X_{i}^{\top }%
\overline{\theta }_{0}+C_{i}\overline{\kappa }_{0})\Big)W_i +A_0^{-1}A_{1}B^{-1}\frac{1}{\sqrt n}\sum_{i=1}^n S_i + o_p(1).
\end{align*}%
Therefore, 
\begin{equation*}
\frac{1}{\sqrt{n}}\sum_{i=1}^{n}(Y_{i}-T_{i}\beta _{A\Lambda
})(Z_{i}-\Lambda (X_{i}^{\top }\widehat{\theta }+\widehat{C}_{i}\widehat{%
\kappa }))=\frac{1}{\sqrt{n}}\sum_{i=1}^{n}(\ell _{i,1}^{A\Lambda }-\ell
_{i,2}^{A\Lambda })+o_{p}(1).
\end{equation*}%
Also, 
\begin{equation*}
\frac{1}{n}\sum_{i=1}^{n}T_{i}(Z_{i}-\Lambda (X_{i}^{\top }\widehat{\theta }+%
\widehat{C}_{i}\widehat{\kappa }))\rightarrow _{p}\mathbb{E}[T(Z-\Lambda
(X^{\top }\overline{\theta }_{0}+C\overline{\kappa }_{0}))].
\end{equation*}%
Hence, 
\begin{align*}
\sqrt{n}(\widehat{\beta }_{A\Lambda }-\beta _{A\Lambda })& =\frac{%
n^{-1/2}\sum_{i=1}^{n}(Y_{i}-T_{i}\beta _{A\Lambda })(Z_{i}-\Lambda
(X_{i}^{\top }\widehat{\theta }+\widehat{C}_{i}\widehat{\kappa }))}{%
n^{-1}\sum_{i=1}^{n}T_{i}(Z_{i}-\Lambda (X_{i}^{\top }\widehat{\theta }+%
\widehat{C}_{i}\widehat{\kappa }))} \\
& =\frac{n^{-1/2}\sum_{i=1}^{n}(\ell _{i,1}^{A\Lambda }-\ell
_{i,2}^{A\Lambda })}{\mathbb{E}[T(Z-\Lambda (X^{\top }\overline{\theta }%
_{0}+C\overline{\kappa }_{0}))]}+o_{p}(1)\rightarrow _{d}N(0,\sigma
_{A\Lambda }^{2}),
\end{align*}%
yielding the asserted claim.
\end{proof}

\begin{section}{Auxiliary Lemma}
\begin{lemma}\label{lem: simple lemma}
Suppose that Assumption \ref{as: monotonicity} is satisfied and that $\mathbb E[|Y|]<\infty$. Then
\begin{align*}
&\mathbb E[(Y(1) - Y(0))(T(1) - T(0))|X] = \Delta_{CP}(X)\omega_{CP}(X),\\
&\mathbb E[(Y(1) - Y(0))T(0)|X]  = \Delta_{AT}(X)\omega_{AT}(X),\\
&\mathbb E[(Y(1) - Y(0))(1 - T(1))|X] = \Delta_{NT}(X)\omega_{NT}(X).
\end{align*}
\end{lemma}
\begin{proof}
Let $\Delta Y = Y(1) - Y(0)$ and $\Delta T = T(1) - T(0)$. Then by the law of iterated expectations and Assumption \ref{as: monotonicity},
$$
\mathbb E[\Delta Y \Delta T|X] = \mathbb E[\Delta Y|X,\Delta T = 1]\mathbb P(\Delta T = 1|X) = \Delta_{CP}(X)\omega_{CP}(X),
$$
$$
\mathbb E[\Delta Y T(0)|X] = \mathbb E[\Delta Y|X,T(0) = 1]\mathbb P(T(0) = 1|X) = \Delta_{AT}(X)\omega_{AT}(X),
$$
$$
\mathbb E[\Delta Y (1-T(1))|X] = \mathbb E[\Delta Y|X,T(1) = 0]\mathbb P(T(1) = 0|X) = \Delta_{NT}(X)\omega_{NT}(X).
$$
The asserted claims follow.
\end{proof}
\end{section}



\bibliographystyle{econ-econometrica}
\bibliography{ref}

\end{document}